\documentclass{llncs}

\usepackage{url}
\usepackage{amsmath}
\usepackage{amssymb}
\usepackage{times}
\usepackage{paralist}
\usepackage{mathrsfs}
\usepackage{listings}
\usepackage{tikz}
\usetikzlibrary{positioning}
\usepackage{graphicx}
\usepackage{wrapfig}

\usepackage{multicol}

\usepackage{algorithm}
\usepackage[noend]{algpseudocode}
\usepackage{macros}

\makeatletter

\begingroup \catcode `|=0 \catcode `[= 1
\catcode`]=2 \catcode `\{=12 \catcode `\}=12
\catcode`\\=12 |gdef|@xcomment#1\end{comment}[|end[comment]]
|endgroup

\def\@comment{\let\do\@makeother \dospecials\catcode`\^^M=10\def\par{}}

\def\begincomment{\@comment\@xcomment}

\makeatother

\newcommand{\myspace}{\vspace*{-0.5em}}

\begin{document}

%
%
%

\title{Strategy Representation by Decision Trees in~Reactive~Synthesis}

\author{Tom\'{a}\v{s} Br\'{a}zdil$^1$, Krishnendu Chatterjee$^2$, Jan K\v{r}et\'{i}nsk\'{y}$^3$, and Viktor Toman$^2$}
\institute{$^1$ Masaryk University, Brno, Czech Republic \\
              $^2$ Institute of Science and Technology Austria \\
              $^3$ Technical University of Munich, Germany\vspace{-2ex}}

\maketitle

\begin{abstract}
Graph games played by two players over finite-state graphs are central
in many problems in computer science. 
In particular, graph games with $\omega$-regular winning conditions, 
specified as parity objectives, which can express properties such as safety, 
liveness, fairness, are the basic framework for verification and synthesis 
of reactive systems.
The decisions for a player at various states of the graph game are represented 
as strategies. 
While the algorithmic problem for solving graph games with parity objectives
has been widely studied, the most prominent data-structure for strategy representation
in graph games has been binary decision diagrams (BDDs).
However, due to the bit-level representation, BDDs do not retain the inherent
flavor of the decisions of strategies, and are notoriously hard to minimize to 
obtain succinct representation.
In this work we propose decision trees for strategy representation in graph games.
Decision trees retain the flavor of decisions of strategies and allow entropy-based
minimization to obtain succinct trees. 
However, decision trees work in settings (e.g., probabilistic models) where errors 
are allowed, and overfitting of data is typically avoided.
In contrast, for strategies in graph games no error is allowed, and the decision 
tree must represent the entire strategy.
We develop new techniques to extend decision trees to overcome the above obstacles,
while retaining the entropy-based techniques to obtain succinct trees.
We have implemented our techniques to extend the existing decision tree solvers.
We present experimental results for problems in reactive synthesis to show
that decision trees provide a much more efficient data-structure for strategy 
representation as compared to BDDs.
\end{abstract}

\section{Introduction}

\smallskip\noindent{\em Graph games.}
We consider nonterminating two-player graph games played 
on finite-state graphs.
The vertices of the graph are partitioned into states controlled by 
the two players, namely, player~1 and player~2, respectively. 
In each round the state changes according to a transition 
chosen by the player controlling the current state.
Thus, the outcome of the game being played for an infinite 
number of rounds, is an infinite path through the graph, which is 
called a play.
An objective for a player specifies whether the resulting play 
is either winning or losing.
We consider zero-sum games where the objectives of the players are 
complementary.
A strategy for a player is a recipe to specify the choice of the 
transitions for states controlled by the player.
Given an objective, a winning strategy for a player from a state ensures 
the objective irrespective of the strategy of the opponent.

\smallskip\noindent{\em Games and synthesis.}
These games play a central role in several areas of computer science.
One important application arises when the vertices and edges of a graph 
represent the states and transitions of a reactive system, and the two 
players represent controllable versus uncontrollable decisions during the
execution of the system.
The \emph{synthesis} problem for reactive systems asks for the construction 
of a winning strategy in the corresponding graph game.
This problem was first posed independently by Church~\cite{Church62} 
and B\"uchi~\cite{Buchi62}, and has been extensively 
studied~\cite{Rabin69,BuchiLandweber69,GH82,McNaughton93}.
Other than applications in synthesis of discrete-event and reactive 
systems~\cite{RamadgeWonham87,PnueliRosner89}, game-theoretic formulations 
play a crucial role in modeling~\cite{Dill89book,ALW89}, refinement~\cite{FairSimulation}, 
verification~\cite{DetectingErrorsBeforeReaching,AHK02}, testing~\cite{GurevichTest}, 
compatibility checking \cite{InterfaceAutomata}, 
and many other applications.
In all the above applications, the objectives are $\omega$-regular, 
and the $\omega$-regular sets of infinite paths provide an important and robust 
paradigm for reactive-system specifications~\cite{MannaPnueliVol1,Thomas97}.

\smallskip\noindent{\em Parity games.}
Graph games with parity objectives are relevant in reactive synthesis, 
since all common specifications for reactive systems are expressed as $\omega$-regular
objectives that can be transformed to parity objectives.
In particular, a convenient specification formalism in reactive synthesis 
is LTL (linear-time temporal logic). 
The LTL synthesis problem asks, given a specification over input and output variables in LTL,
whether there is a strategy for the output sequences to ensure the specification irrespective of the
behavior of the input sequences.
The conversion of LTL to non-deterministic B\"uchi automata, and non-deterministic 
B\"uchi automata to deterministic parity automata, gives rise to a parity game to solve
the LTL synthesis problem.
Formally, the algorithmic problem asks for a given graph game with a parity objective 
and a starting state, whether player~1 has a winning strategy. This problem is  central in verification and synthesis.
While it is a major open problem whether the problem can be solved in polynomial 
time, it has been widely studied in the literature~\cite{Zie98,CaludeJKLS17,Schewe17}.

\smallskip\noindent{\em Strategy representation.}
In graph games, the strategies are the most important objects as they 
represent the witness to winning of a player.
For example, winning strategies represent controllers in the controller synthesis problem.
Hence all parity-games solvers produce the winning strategies as their output.
While the algorithmic problem of solving parity games has received huge attention,
quite surprisingly, data-structures for representation of strategies have 
received little attention.
While the data-structures for strategies could be relevant in particular algorithms for parity games 
(e.g., strategy-iteration algorithm), our focus is very different than improving such algorithms.
Our main focus is the representation of the strategies themselves, which are the main output 
of the parity-games solvers, and hence our strategy representation serves as post-processing
of the output of the solvers.
The standard data-structure for representing strategies is binary decision diagrams
(BDDs)~\cite{Akers78,Bryant86} and it is used as follows:
a strategy is interpreted as a lookup table of pairs that specifies for every controlled 
state of the player the transition to choose, and then the lookup table is 
represented as a binary decision diagram (BDD).

\smallskip\noindent{\em Strategies as BDDs.}
The desired properties of data-structures for strategies are as follows: 
(a)~{\em succinctness}, i.e., small strategies are desirable, since strategies 
correspond to controllers, and smaller strategies represent efficient controllers 
that are required in resource-constrained environments such as embedded systems; 
(b)~{\em explanatory}, i.e., the representation explains the decisions of 
the strategies. 
In this work we consider different data-structure for representation of strategies 
in graph games. 
The key drawbacks of BDDs to represent strategies in graph games are as follows.
First, the size of BDDs crucially depends on the variable ordering.
The variable ordering problem is notoriously difficult:  
the optimal variable ordering problem is NP-complete, and for large dimensions 
no heuristics are known to work well.
Second, due to the fact that strategies have to be input to the BDD construction
as Boolean formulae, the representation though succinct, does not retain
the inherent important choice features of the decisions of the strategies (for 
an illustration see Example~\ref{ex:stratdt}).

\smallskip\noindent{\em Strategies as decision trees.}
In this work, we propose to use {\em decision trees}, i.e.~\cite{Mitchell1997}, 
for strategy representation in graph games. 
A decision tree is a structure similar to a BDD, but with nodes labelled by various 
predicates over the system's variables.
In the basic algorithm for decision trees, the tree is constructed using an 
unfolding procedure where the branching for the decision making is done in 
order to maximize the information gain at each step. 

The key advantages of decision trees over BDDs are as follows:
\begin{compactitem} 
\item The first two advantages are conceptual.
First, while in BDDs, a level corresponds to one variable, in decision trees,
a predicate can appear at different levels and different predicates 
can appear at the same level.
This allows for more flexibility in the representation.
Second, decision trees utilize various predicates over the given features
in order to make decisions, and ignore all the unimportant features.
Thus they retain the inherent flavor of the decisions of the strategies. 

\item The other important advantage is algorithmic.
Since the data-structure is based on information gain, sophisticated 
algorithms based on entropy exist for their construction. These algorithms 
result in a succinct representation, whereas for BDDs there is no 
good algorithmic approach for variable reordering.

\end{compactitem}

\smallskip\noindent{\em Key challenges.}
While there are several advantages of decision trees, and decision trees have been 
extensively studied in the machine learning community, there are several 
key challenges and obstacles for representation of strategies in graph games 
by decision trees. 
\begin{compactitem}
\item First, decision trees have been mainly used in the probabilistic 
setting. In such settings, research from the machine learning community 
has developed techniques to show that decision trees can be effectively 
pruned to obtain succinct trees, while allowing small error probabilities.
However, in the context of graph games, no error is allowed in the strategic choices.

\item Second, decision trees have been used in the machine learning 
community in classification, where an important aspect is to ensure 
that there is no overfitting of the training data. 
In contrast, in the context of graph games, the decision tree must fit the 
entire representation of the strategies.
\end{compactitem}
While for probabilistic models such as Markov decision processes (MDPs),
decision trees can be used as a blackbox~\cite{DBLP:conf/cav/BrazdilCCFK15}, 
in the setting of graph games their use is much more challenging. 
In summary, in previous settings where decision trees are used 
small error rates are allowed in favor of succinctness, and overfitting is not
permitted, whereas in our setting no error is allowed, and the complete fitting of the 
tree has to be ensured.
The basic algorithm for decision-tree learning (called ID3 algorithm \cite{DBLP:journals/ml/Quinlan86,Mitchell1997})
suffers from  the curse of dimensionality, and the error allowance is used to handle the dimensionality. 
Hence we need to develop new techniques for strategy learning with decision trees in graph 
games.

\smallskip\noindent{\em Our techniques.}
We present a new technique for learning strategies with decision trees based on 
{\em look-ahead}.
In the basic algorithm for decision trees, at each step of the 
unfolding, the algorithm proceeds as long as there is any information gain.
However, suppose for no possible branching there is any information gain.
This represents the situation where the local (i.e., one-step based) decision making fails 
to achieve information gain. 
We extend this process so that look-ahead is allowed, i.e., we consider possible 
information gain with multiple steps.
The look-ahead along with complete unfolding ensure that there is no error in the strategy representation.
While the look-ahead approach provides a systematic principle to 
obtain precise strategy representation, it is computationally expensive, and we 
present heuristics used together with look-ahead for computational efficiency 
and succinctness of strategy representation.

\smallskip\noindent{\em Implementation and experimental results.}
Since in our setting existing decision tree solvers cannot be used as a blackbox,
we extended the existing solvers with our techniques mentioned above. 
We have then applied our implementation to compare decision trees and 
BDDs for representation of strategies for problems in reactive synthesis.
First, we compared our approach against BDDs for two classical examples
of reactive synthesis from SYNTCOMP benchmarks~\cite{DBLP:journals/corr/JacobsBBKPRRSST16}.
Second, we considered randomly generated LTL formulae, and the graph 
games obtained for the realizability of such formulae. 
In both the above experiments the decision trees represent the winning 
strategies much more efficiently as compared to BDDs.

\smallskip\noindent{\em Related work.}
Previous non-explicit representation of strategies for verification or synthesis purposes typically used BDDs~\cite{WBB+10} or automata~\cite{DBLP:conf/atva/Neider11,DBLP:conf/tacas/NeiderT16} 
and do not explain the decisions by the current valuation of variables.
\emph{Decision trees} have been used a lot in the area of machine learning as a classifier that naturally explains a decision \cite{Mitchell1997}.
They have also been considered 
for approximate representation of values in states and thus implicitly for an approximate representation of \emph{strategies}, 
for the model of Markov decision processes (MDPs)~in \cite{DBLP:conf/ijcai/BoutilierDG95,DBLP:conf/icml/BoutilierD96}.
Recently, in the context of verification, this approach has been modified to capture strategies guaranteed to be $\varepsilon$-optimal, for MDPs~\cite{DBLP:conf/cav/BrazdilCCFK15} 
and partially observable MDPs~\cite{BCCGN16}.
Learning a compact decision tree representation of an MDP strategy was also investigated in~\cite{LPRT10} for the case of body sensor networks.
Besides, decision trees are becoming more popular in verification and programming languages in general, for instance, 
they are used to capture program invariants~\cite{DBLP:journals/corr/KrishnaPW15,DBLP:conf/popl/0001NMR16}.
To the best of our knowledge, decision trees were only used in the context of (possibly probabilistic) systems with only a single player. 
Our decision-tree approach is thus the first in the game setting with two players that is required in reactive synthesis.

\smallskip\noindent{\em Summary.}
To summarize, our main contributions are:
\begin{compactenum}
\item We propose decision trees as data-structure for strategy representation 
in graph games.

\item The representation of strategies with decision trees poses many obstacles,
as in contrast to the probabilistic setting no error is allowed in games.
We present techniques that overcome these obstacles while still retaining the 
algorithmic advantages (such as entropy-based methods) of decision trees 
to obtain succinct decision trees.  

\item We extend existing decision tree solvers with our techniques and present
experimental results to demonstrate the effectiveness of our approach in reactive
synthesis.
 
\end{compactenum}

\newcommand{\stp}{\mathcal{V}}

\section{Graph Games and Strategies}\label{sec:prelim}

\smallskip\noindent{\bf Graph games.}
A {\em graph game} consists of a tuple $G=\tuple{S,S_1,S_2,A_1,A_2,\trans}$, 
where:
\vspace{-1mm}
\begin{itemize}
\item $S$ is a finite set of states partitioned into player 1 states $S_1$ 
and player 2 states $S_2$; 
\item $A_1$ (resp., $A_2$) is the set of actions for 
player~1 (resp., player~2);  and 
\item $\trans \colon (S_1 \times A_1) \cup (S_2 \times A_2) \to S$ is the transition function 
that given a player 1 state and a player 1 action, or a player 2 state and 
a player 2 action, gives the successor state.
\end{itemize}
\vspace{-1mm}

\smallskip\noindent{\bf Plays.}
A \emph{play} is an infinite sequence of state-action pairs 
$\tuple{s_0 a_0 s_1 a_1 \ldots}$ such that for all $j \geq 0$ we have that 
if $s_j \in S_i$ for $i\in \set{1,2}$, then $a_j \in A_i$ and 
$\trans(s_j,a_j)=s_{j+1}$.
We denote by $\Plays(G)$ the set of all plays of a graph game $G$.

\smallskip\noindent{\bf Strategies.} 
A strategy is a recipe for a player to choose actions to extend finite prefixes of plays.
Formally, a strategy $\straa$ for player~1 is a function 
$\straa \colon S^\star \cdot S_1 \to A_1$ that given a finite sequence of visited states 
chooses the next action.
The definitions for player~2 strategies $\strab$ are analogous.
We denote by $\Straa(G)$ and $\Strab(G)$ the set of all strategies for player~1 and 
player~2 in graph game $G$, respectively.
Given strategies $\straa \in \Straa(G)$ and $\strab \in \Strab(G)$, and a starting 
state $s$ in $G$, there is a unique play $\pat(s,\straa,\strab)=\tuple{s_0 a_0 s_1 a_1 \ldots}$
such that $s_0=s$ and for all $j \geq 0$ if $s_j \in S_1$ (resp., $s_j \in S_2$) 
then $a_j=\straa(\tuple{s_0  s_1 \ldots s_j})$
(resp.,  $a_j=\strab(\tuple{s_0 s_1 \ldots s_j})$).
A {\em memoryless} strategy is a strategy that does not depend on the 
finite prefix of the play but only on the current state, i.e., 
functions $\straa \colon S_1 \to A_1$ and $\strab \colon S_2 \to A_2$.

\smallskip\noindent{\bf Objectives.} An \emph{objective} for
a graph game $G$ is a set $\varphi \subseteq \Plays(G)$.
We consider the following objectives:
\begin{itemize}
\item \emph{Reachability and safety objectives.} 
A reachability objective is defined by a set $T \subseteq S$ of target states, and 
the objective requires that a state in $T$ is visited at least once. 
Formally, 
$\Reach(F)= \set{\tuple{s_0 a_0 s_1 a_1 \ldots} \in \Plays(G) \mid \exists i :\, s_i\in T}$.  
The dual of reachability objectives are safety objectives, 
defined by a set $F \subseteq S$ of safe states, and the objective requires that only 
states in $F$ are visited.  
Formally, 
$\Safe(F)=\set{\tuple{s_0 a_0 s_1 a_1 \ldots} \in \Plays(G) \mid \forall i :\, s_i \in F}$.

\item \emph{Parity objectives.} 
For an infinite play $\pat$ we denote by $\Inf(\pat)$ the set of states that 
occur infinitely often in $\pat$. 
Let $p \colon S \to \Nats$ be a \emph{priority function}.  
The \emph{parity} objective $\Parity(p) = \{\pat \in \Plays(G) \mid \min\{p(s) 
\mid s \in \Inf(\pat)\} \text{ is even }\}$ requires that the minimum of the 
priorities of the states visited infinitely often be even. 
\end{itemize}

\smallskip\noindent{\bf Winning region and strategies.} 
Given a game graph $G$ and an objective $\varphi$, a {\em winning} strategy
$\straa$ from state $s$ for player~1 is a strategy such that for all 
strategies $\strab \in \Strab(G)$ we have  $\pat(s,\straa,\strab) \in \varphi$.
Analogously, a winning strategy $\strab$ for player~2 from $s$ ensures that
for all strategies $\straa \in \Straa(G)$ we have $\pat(s,\straa,\strab) 
\not\in \varphi$.
The {\em winning region} $W_1(G,\varphi)$ (resp., $W_2(G,\overline{\varphi})$) 
for player~1 (resp., player~2) is the set of states such that player~1 (resp., player~2) 
has a winning strategy.
A fundamental result for graph games with parity objectives shows that
the winning regions form a partition of the state space, and if there is 
a winning strategy for a player, then there is a memoryless winning strategy~\cite{EJ91}.

\smallskip\noindent{\bf LTL synthesis and objectives.}
Reachability and safety objectives are the most basic objectives to specify 
properties of reactive systems.
Most properties that arise in practice for analysis of reactive systems are 
$\omega$-regular objectives. A convenient logical framework to express 
$\omega$-regular objectives is the LTL (linear-time temporal logic) framework.
The problem of synthesis from specifications, in particular, LTL synthesis has
received huge attention~\cite{ModelCheckBook}. 
LTL objectives can be translated to parity automata, and the synthesis problem 
reduces to solving games with parity objectives.

\smallskip

In reactive synthesis it is natural to consider games where the state space
is defined by a set of variables, and the game is played by input and 
output player who choose the respective input and output signals.
We describe such games below that easily correspond to graph games.

\smallskip\noindent{\bf I/O games with variables.}
Consider a finite set $X=\set{x_1,x_2,\ldots,x_n}$ of variables from a finite domain; for simplicity, we consider Boolean variables only.
A {\em valuation} is an assignment to each variable, in our case $2^X$
denotes the set of all valuations.
Let $X$ be partitioned into input signals, output signals, and state variables, i.e., $X=I\uplus O\uplus V$.
Consider the alphabet $\inp=2^{I}$ (resp., $\oup=2^{O}$) where each letter represents
a subset of the input (resp., output) signals and the alphabet $\stp=2^V$ where each letter represents a subset of state variables.
The input/output choices affect the valuation of the variables, which is given by 
the next-step valuation function $\Trans \colon \stp \times \inp \times \oup \to \stp$.
Consider a game played as follows: at every round the input player chooses
a set of input signals (i.e., a letter from $\inp$), and given the input choice the
output player chooses a set of output signals (i.e., a letter from $\oup$). 
The above game can be represented as a graph game $\tuple{S,S_1,S_2,A_1,A_2,\trans}$ 
as follows:
\begin{itemize}
\item $S=\stp \cup (\stp\times \inp)$;
\item player~1 represents the input player and $S_1=\stp$;
player~2 represents the output player and $S_2=\stp \times \inp$;
\item $A_1=\inp$ and $A_2=\oup$; and 
\item given a valuation $v \in \stp$ and $a_1 \in A_1$ we have $\trans(v,a_1)=(v,a_1)$,
and for $a_2 \in A_2$ we have $\trans((v,a_1),a_2)=\Trans(v,a_1,a_2)$.
\end{itemize}

In this paper, we use decision trees to represent memoryless strategies in such graph games, where states are
represented as vectors of Boolean values. In Section \ref{sec:exper} we show how such games arise from various
sources (AIGER specifications~\cite{AIGER}, LTL synthesis) and why it is sufficient to consider memoryless strategies only.

\newcommand{\good}{\mathit{YES}}
\newcommand{\bad}{\mathit{NO}}

\section{Decision Trees and Decision Tree Learning}\label{sec:repr}
In this section we recall decision trees and learning decision trees. 
A key application domain of games on graphs is reactive synthesis (such as safety synthesis from SYNTCOMP benchmarks as well 
as LTL synthesis) and our comparison for strategy representation is against BDDs.
BDDs are particularly suitable for states and actions represented as bitvectors.
Hence for a fair comparison against BDDs, we consider a simple version of decision trees over bitvectors,
though decision trees and their corresponding methods can be naturally extended to richer domains 
(such as vectors of integers as used in~\cite{DBLP:conf/cav/BrazdilCCFK15}).

\subsubsection{Decision trees.}
A \emph{decision tree} over $\{0,1\}^d$ is a tuple $\mathcal{T}=(T,\rho,\theta)$ where $T$ is a finite rooted binary (ordered) tree
with a set of inner nodes $N$ and a~set of leaves $L$, $\rho$ assigns to every inner node a number of $\{1,\ldots,d\}$, and $\theta$
assigns to every leaf a value $\good$ or $\bad$.

The language $\mathcal{L}(\mathcal{T})\subseteq \{0,1\}^d$ of the tree is defined as follows. For a vector $\vec{x}=(x_1,\ldots, x_d)\in \{0,1\}^d$,
we find a path $p$ from the root to a leaf such that for each inner node $n$ on the path, $\vec{x}(\rho(n))=0$ iff the~first child of $n$ is on $p$.
Denote the leaf on this particular path by $\ell$. Then $\vec{x}$ is in the language $\mathcal{L}(\mathcal{T})$ of $\mathcal T$ iff $\theta(\ell)=\good$.

\begin{example}
Consider dimension $d=3$. The language of the tree depicted in Fig.~\ref{fig:dectree} can be described by the following regular
expression $\{0,1\}^2\cdot 0 + \{0,1\}\cdot 1\cdot 1$. Intuitively, the root node represents the predicate of the third value, the other
inner node represents the predicate of the second value. For each inner node, the first and second children correspond to the cases
where the value at the position specified by the predicate of the inner node is $0$ and $1$, respectively. We supply the edge labels to
depict the tree clearly. The leftmost leaf corresponds to the subset of $\{0,1\}^3$ where the third value is $0$, the rightmost leaf
corresponds to the subset of $\{0,1\}^3$ where the third value is $1$ and the second value is $1$.
\end{example}

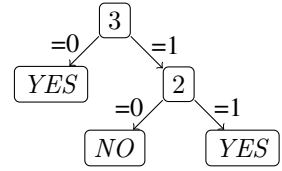
\begin{wrapfigure}{r}{0.38\textwidth}
\begin{center}
\begin{tikzpicture}[node distance = 1.2cm]
	\node (a1) [draw,rectangle,rounded corners=2pt] {$3$} ;
	\node (a2) [below right of=a1,draw,rectangle,rounded corners=2pt]  {$2$};
	\node (good1) [below right of=a2,draw,rectangle,rounded corners=2pt] {$\good$};
	\node (bad1) [below left of=a2,draw,rectangle,rounded corners=2pt] {$\bad$};
	\node (good2) [below left of=a1,draw,rectangle,rounded corners=2pt] {$\good$};
	
	\draw [->] (a1) -- node [xshift=0.25cm, yshift=0.1cm] {=1} (a2);
	\draw [->] (a1) -- node [xshift=-0.25cm, yshift=0.1cm] {=0} (good2);
	\draw [->] (a2) -- node [xshift=0.25cm, yshift=0.1cm] {=1} (good1);
	\draw [->] (a2) -- node [xshift=-0.25cm, yshift=0.1cm] {=0} (bad1);
\end{tikzpicture}

	\caption{A decision tree over $\{0,1\}^3$}
	\label{fig:dectree}
\end{center}
\end{wrapfigure}

\setlength{\belowdisplayskip}{4pt} \setlength{\abovedisplayskip}{4pt}
\smallskip\noindent{\bf Standard DT learning.}
We describe the standard process of binary classification using decision trees (see Algorithm~\ref{alg:ID3}).
Given a {\em training set} $\mathit{Train}\subseteq \{0,1\}^d$, partitioned into two subsets $\mathit{Good}$ and $\mathit{Bad}$,
the~process of learning according to the algorithm ID3 \cite{DBLP:journals/ml/Quinlan86,Mitchell1997} computes a decision
tree $\mathcal{T}$ that assigns $\good$ to all elements of $\mathit{Good}$ and $\bad$ to all elements of $\mathit{Bad}$.
In the algorithm, a leaf $\ell\subseteq \{0,1\}^d$ is {\em mixed} if $\ell$ has a non-empty intersection with both $\mathit{Good}$
and $\mathit{Bad}$. To split a leaf $\ell$ on $\mathit{bit}\in \{1,\ldots,d\}$  means that $\ell$ becomes an internal node with the
two new leaves $\ell_0$ and $\ell_1$ as its children. Then, the leaf $\ell_0$ contains the samples of $\ell$ where the value in the
position $\mathit{bit}$ equals $0$, and the leaf $\ell_1$ contains the rest of the samples of $\ell$, since these have the value in
the position $\mathit{bit}$ equal to $1$. The \emph{entropy} of a node is defined as
\[
H(\ell) = - \frac{|\ell\cap\mathit{Good}|}{|\ell|} log_2 \frac{|\ell\cap\mathit{Good}|}{|\ell|} - \frac{|\ell\cap\mathit{Bad}|}{|\ell|} log_2 \frac{|\ell\cap\mathit{Bad}|}{|\ell|}
\]
\setlength{\belowdisplayskip}{0pt} \setlength{\abovedisplayskip}{0pt}
An {\em information gain} of a given $\mathit{bit}\in \{1,\ldots,d\}$ (and thus also of the split into $\ell_0$ and $\ell_1$) is defined by
			\begin{equation}
			\label{eq:information_gain}
			H(\ell) - \frac{|\ell_0|}{|\ell|} H(\ell_0) - \frac{|\ell_1|}{|\ell|} H(\ell_1)
			\end{equation}
			where $\ell_0$ is the set of all $\vec{x}=(x_1,\ldots,x_{d})\in \ell\subseteq \{0,1\}^{d}$ with $x_{\mathit{bit}}=0$ and $\ell_1=\ell\smallsetminus \ell_0$. Finally, given $\ell\subseteq \{0,1\}^d$ we define
			\[
			\mathit{maxclass}(\ell) = \begin{cases}
				\good & |\ell\cap\mathit{Good}|\geq |\ell\cap\mathit{Bad}|\\
				\bad & \text{otherwise.}
			\end{cases}
			\]
\begin{algorithm}[!t]
	\caption{ID3 learning algorithm\label{alg:ID3}}
	\begin{algorithmic}[1]
		\Statex \textbf{Inputs:} $\mathit{Train}\subseteq\{0,1\}^{d}$ partitioned into subsets $\mathit{Good}$ and $\mathit{Bad}$.
		\Statex \textbf{Outputs:} A decision tree $\mathcal{T}$ such that $\mathcal{L}(\mathcal{T})\cap\mathit{Train}=\mathit{Good}$.
		\Statex /* train $\mathcal{T}$ on positive set $\mathit{Good}$ and negative set $\mathit{Bad}$ */
		\State $\mathcal{T}\gets(\{\mathit{Train}\},\emptyset,\{\mathit{Train}\mapsto^\theta\good\})$ 
		\While{a mixed leaf $\ell$ exists}  
		\State $\mathit{bit}\gets$ an element of $\{1,\ldots,d\}$ that maximizes the information gain
		\State split $\ell$ on $\mathit{bit}$ into two leaves $\ell_0$ and $\ell_1$, $\rho(\ell) = \mathit{bit}$
		\State $\theta(\ell_0)\gets maxclass(\ell_0)$ and $\theta(\ell_1)\gets maxclass(\ell_1)$
		\EndWhile	
		\State\Return $\mathcal{T}$
	\end{algorithmic}
\end{algorithm}
\setlength{\belowdisplayskip}{6pt} \setlength{\abovedisplayskip}{8pt}

Intuitively, the splitting on the component with the highest gain splits the set so that it maximizes the
portion of $\mathit{Good}$ in one subset and the portion of $\mathit{Bad}$ in the other one.

\begin{remark}[Optimizations]\label{rem:opt}
The basic ID3 algorithm for decision tree learning suffers from the curse of dimensionality. 
However, decision trees are primarily applied to machine learning problems where small errors
are allowed to obtain succinct trees. 
Hence the allowance of error is crucially used in existing solvers (such as WEKA~\cite{WEKA})
to combat dimensionality.
In particular, the error rate is exploited in the unfolding, where the unfolding proceeds 
only when the information gain exceeds the error threshold.
Further error is also introduced in the pruning of the trees, which ensures that the overfitting of training data
is avoided.
\end{remark}

\section{Learning Winning Strategies Efficiently}\label{sec:learn}
In this section we present our contributions. 
We first start with the representation of strategies as training sets,
and then present our strategy decision-tree learning algorithm.

\subsection{Strategies as Training Sets and Decision Trees}

\noindent{\bf Strategies as training sets.}
Let us consider a game $G=\tuple{S,S_1,S_2,A_1,A_2,\trans}$. We represent strategies of both players using the same method.
So in what follows we consider either of the players and denote by $S_{*}$ and $A_{*}$ the sets of states and actions
of the player, respectively. We fix $\tilde{\sigma} \colon S_*\rightarrow A_*$, a memoryless strategy of the player. 

We assume that $G$ is an I/O game with binary variables, which means $S_*\subseteq \{0,1\}^{n}$ and $A_*\subseteq \{0,1\}^{a}$.
A memoryless strategy is then a partial function \mbox{$\tilde{\sigma} \colon \{0,1\}^{n} \rightarrow \{0,1\}^{a}$}. Furthermore, we fix an initial
state $s_0$, and let $S^R_* \subseteq \{0,1\}^{n}$ be the set of all states reachable from $s_0$ using $\sigma$ against some
strategy of the other player. We consider all objectives only on plays starting in the initial state $s_0$. Therefore, the strategy
can be seen as a function $\sigma \colon S^R_*\rightarrow A_*$ such that $\sigma=\tilde{\sigma}|_{S^R_*}$.

\vbox{
Now we define
\begin{compactitem}
\item $\mathit{Good} = \{\tuple{s,\sigma(s)}\in S^R_*\times A_*\}$
\item $\mathit{Bad} = \{\tuple{s,a} \in S^R_*\times A_* \mid a\not = \sigma(s)\}$
\end{compactitem}
The set of all training examples is a disjunctive union $\mathit{Train} = \mathit{Good} \uplus \mathit{Bad} \subseteq \{0,1\}^{n+a}$.
}

As we do not use any pruning or stopping rules, the ID3 
algorithm returns a decision tree $\mathcal{T}$ that fits the training
set $\mathit{Train}$ exactly. This means that for all $s\in S^R_*$ we have that $\tuple{s,a}\in \mathcal{L}(\mathcal{T})$ iff $\sigma(s)=a$.
Thus $\mathcal{T}$ represents the strategy $\sigma$. Note that for any sample of $\{0,1\}^{n+a} \setminus \mathit{Train}$, the fact
whether it belongs to $\mathcal{L}(\mathcal{T})$ or not is immaterial to us.
Thus strategies are naturally represented as decision trees, and we present an illustration below.

\setlength{\intextsep}{2pt}
\setlength{\abovecaptionskip}{0pt}
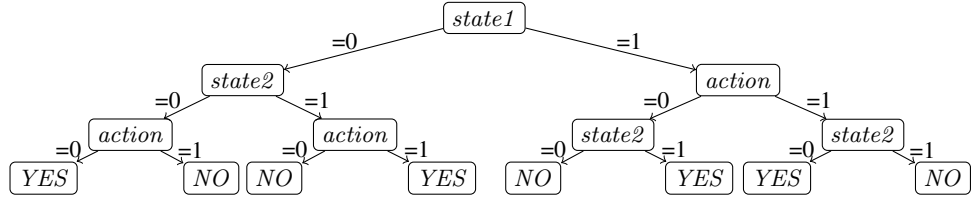
\begin{figure}
\begin{center}
\vspace{4pt}
\begin{tikzpicture}[]
	\node (a1) [draw,rectangle,rounded corners=2pt] {$\mathit{state1}$};
    \node (dummy) [below=0.5cm of a1] {};
    \node (a2) [left=2.55cm of dummy,draw,rectangle,rounded corners=2pt] {$\mathit{state2}$};
    \node (dummy2) [below=0.4cm of a2] {};
    \node (a3) [right=2.7cm of dummy,draw,rectangle,rounded corners=2pt]  {$\mathit{action}$};
    \node (dummy3) [below=0.4cm of a3] {};
    \node (a4) [left=0.82cm of dummy2,draw,rectangle,rounded corners=2pt] {$\mathit{action}$};
	\node (a5) [right=0.82cm of dummy2,draw,rectangle,rounded corners=2pt] {$\mathit{action}$};
	\node (a6) [left=1cm of dummy3,draw,rectangle,rounded corners=2pt] {$\mathit{state2}$};
    \node (a7) [right=1cm of dummy3,draw,rectangle,rounded corners=2pt] {$\mathit{state2}$};
	
    \node (good1) [below left=0.2cm of a4,draw,rectangle,rounded corners=2pt] {$\good$};
    \node (bad1) [below right=0.2cm of a4,draw,rectangle,rounded corners=2pt] {$\bad$};
	
    \node (good2) [below right=0.2cm of a5,draw,rectangle,rounded corners=2pt] {$\good$};
    \node (bad2) [below left=0.2cm of a5,draw,rectangle,rounded corners=2pt] {$\bad$};
	
    \node (good3) [below right=0.2cm of a6,draw,rectangle,rounded corners=2pt] {$\good$};
    \node (bad3) [below left=0.2cm of a6,draw,rectangle,rounded corners=2pt] {$\bad$};
	
    \node (good4) [below left=0.2cm of a7,draw,rectangle,rounded corners=2pt] {$\good$};
    \node (bad4) [below right=0.2cm of a7,draw,rectangle,rounded corners=2pt] {$\bad$};

    \draw [->] (a1) -- node [xshift=-0.25cm, yshift=0.1cm] {=0} (a2);
	\draw [->] (a1) -- node [xshift=0.25cm, yshift=0.1cm] {=1} (a3);
	
    \draw [->] (a2) -- node [xshift=-0.25cm, yshift=0.1cm] {=0} (a4);
	\draw [->] (a2) -- node [xshift=0.25cm, yshift=0.1cm] {=1} (a5);
	
	\draw [->] (a3) -- node [xshift=-0.25cm, yshift=0.1cm] {=0} (a6);
    \draw [->] (a3) -- node [xshift=0.25cm, yshift=0.1cm] {=1} (a7);
	
    \draw [->] (a4) -- node [xshift=-0.25cm, yshift=0.1cm] {=0} (good1);
    \draw [->] (a4) -- node [xshift=0.25cm, yshift=0.1cm] {=1} (bad1);
	
    \draw [->] (a5) -- node [xshift=-0.25cm, yshift=0.1cm] {=0} (bad2);
    \draw [->] (a5) -- node [xshift=0.25cm, yshift=0.1cm] {=1} (good2);
	
    \draw [->] (a6) -- node [xshift=-0.25cm, yshift=0.1cm] {=0} (bad3);
    \draw [->] (a6) -- node [xshift=0.25cm, yshift=0.1cm] {=1} (good3);
	
    \draw [->] (a7) -- node [xshift=-0.25cm, yshift=0.1cm] {=0} (good4);
    \draw [->] (a7) -- node [xshift=0.25cm, yshift=0.1cm] {=1} (bad4);
\end{tikzpicture}
\end{center}
	\caption{Tree representation of strategy $\sigma$}
	\label{fig:strategyrepr}
\end{figure}

\begin{example}\label{ex:stratdt}
Let the state binary variables be labeled as $\mathit{state1}$, $\mathit{state2}$, and $\mathit{state3}$, respectively,
and let the action binary variable be labeled as $\mathit{action}$. Consider a strategy $\sigma$ such that
\mbox{$\sigma(0,0,0)=0$, $\;\;\sigma(0,1,0)=1$, $\;\;\sigma(1,0,0)=1$, $\;\;\sigma(1,1,1)=0$}. Then 
\begin{compactitem}
\item $\mathit{Good} = \{(0,0,0,0), (0,1,0,1), (1,0,0,1), (1,1,1,0)\}$
\item $\mathit{Bad} = \{(0,0,0,1), (0,1,0,0), (1,0,0,0), (1,1,1,1)\}$
\end{compactitem}
 Fig.~\ref{fig:strategyrepr}
depicts a decision tree $\mathcal{T}$ representing the strategy~$\sigma$.
\end{example}

\begin{remark}
The above example demonstrates the conceptual advantages of decision trees over BDDs. 
First, in decision trees, different predicates can appear at the same level of the tree
(e.g. predicates $\mathit{state2}$ and $\mathit{action}$ appear at the second level). At the same time, a predicate can appear at
different levels of the tree (e.g. predicate $\mathit{action}$ appears once at the second level and twice at the third level).

Second advantage is a bit technical, but very crucial. In the example there is no pair of samples
$g \in \mathit{Good}$ and $b \in \mathit{Bad}$ that differs only in the value of $\emph{state3}$. This suggests that
the feature $\emph{state3}$ is unimportant w.r.t. differentiating between $\mathit{Good}$ and $\mathit{Bad}$, and indeed the decision
tree~$\mathcal{T}$ in Fig.~\ref{fig:strategyrepr} contains no predicate $\emph{state3}$ while still representing~$\sigma$.
However, to construct a BDD that ignores $\emph{state3}$ is very difficult,
since a Boolean formula is provided as the input to the BDD construction,
and this formula inevitably sets the value for every sample.
Therefore, it is impossible to declare
``the samples of $\{0,1\}^{n+a} \setminus \mathit{Train}$ can be resolved either way''. One way to construct a BDD~$\mathcal{B}$
would be $\mathcal{B} \equiv \bigvee_{g \in \mathit{Good}} g$. But then $\mathcal{B}(0,0,0,0)=1$ and $\mathcal{B}(0,0,1,0)=0$, so
$\emph{state3}$ has to be used in the representation of~$\mathcal{B}$. Another option could be
$\mathcal{B} \equiv \bigwedge_{b \in \mathit{Bad}} \neg b$, but then $\mathcal{B}(0,0,0,1)=0$ and $\mathcal{B}(0,0,1,1)=1$,
so $\emph{state3}$ still has to be used in the representation.
\end{remark}

\begin{example}
Consider $\mathit{Good} = \{(0,0,0,0,1)\}$ and $\mathit{Bad} = \{(0,0,0,0,0)\}$.
Algorithm~\ref{alg:ID3}
outputs a simple decision tree differentiating between $\mathit{Good}$ and $\mathit{Bad}$ only according to the value
of the last variable. On the other hand, a BDD constructed as $\mathcal{B} \equiv \bigvee_{g \in \mathit{Good}} g$
contains nodes for all five variables.
\end{example}

\newcommand{\argmax}{\mathop{\mathrm{arg\,max}}}
\newcommand{\argmin}{\mathop{\mathrm{arg\,min}}}

\subsection{Strategy-DT Learning}

\subsubsection{Challenges.} 
In contrast to other machine learning domains, where errors are 
allowed,  since strategies in graph games must be represented precisely, 
several challenges arise.
Most importantly, the machine-learning philosophy of classifiers is to generalize the experience, 
trying to achieve good predictions on any (not just training) data.
In order to do so, overfitting the training data must be avoided. 
Indeed, specializing the classifier to cover the training data precisely leads to classifiers reflecting the concrete instances of random noise instead of generally useful predictors.
Overfitting is prevented using a tolerance on learning all details of the training data.
Consequently, the training data are not learnt exactly.
Since in our case, the training set is exactly what we want to represent, our approach must be entirely different.
In particular, the optimizations in the setting where errors are allowed (see Remark~\ref{rem:opt}) are not applicable
to handle the curse of dimensionality.
In particular, it may be necessary to unfold the decision tree even in situations where none of the one-step unfolds induces any information gain.

\myspace\myspace

\subsubsection{Solution: look-ahead.}
In the ID3 algorithm Alg.~\ref{alg:ID3}, when none of the splits 
has a positive information gain (see Formula~(\ref{eq:information_gain})), the corresponding node is split arbitrarily.
This can result in very large decision trees. We propose a better solution.
Namely, we extend ID3 with a \emph{``look-ahead''}: If no split results in a positive information gain, one can pick a split so that next,
when splitting the children, the information gain is positive. If still no such split exists, one can try and pick a split and splits of
children so that afterwards there is a split of grandchildren with positive information gain. 
And so on, possibly until a constant depth $k$, yielding a \emph{$k$-look-ahead}.

Before we define the look-ahead formally, we have a look at a simple example:

\begin{example}
	Consider $\mathit{Good}=\{(0,0,0,0,0,1,1),(0,0,0,0,0,0,0)\}$ and $\mathit{Bad}=\{(0,0,0,0,0,1,0),(0,0,0,0,0,0,1)\}$, characterising $x_6=x_7$.
	Splitting on any $x_i$, $\;i~\in~\{1,...,7\}$ does not give a positive information gain.
	Standard DT learning procedures would either stop here and not expand this leaf any more, or split arbitrarily.
	With the look-ahead, one can see that using $x_6$ and then $x_7$, the information gain is positive and we obtain a decision tree classifying the set perfectly.
	
	Here we could as well introduce more complex predicates such as $x_6\ \mathtt{xor}\ x_7$ instead of look-ahead.
	However, in general the look-ahead has the advantage that each of the $0$ and $1$ branches may afterwards split on
	different bits (currently best ones), whereas with $x_6\ \mathtt{xor}\ x_7$ we commit to using $x_7$ in both branches. 
\end{example}

The example illustrates the 2-look-ahead with the following formal definition. (For explanatory reasons, the general case follows afterwards.)
Consider a node $\ell\subseteq\{0,1\}^d$.
For every $\mathit{bit},\mathit{bit}_0,\mathit{bit}_1\in \{1,\ldots,d\}$, consider splitting on $\mathit{bit}$ and subsequently the $0$-child
on $\mathit{bit}_0$ and the $1$-child on $\mathit{bit}_1$. This results in a partition
$P(\mathit{bit},\mathit{bit}_0,\mathit{bit}_1)=\{\ell_{00},\ell_{01},\ell_{10},\ell_{11}\}$  of $\ell$. We assign to
$P(\mathit{bit},\mathit{bit}_0,\mathit{bit}_1)$ its {\em 2-look-ahead information gain} defined by\myspace
\begin{align*}
\mathit{IG}(\mathit{bit},& \mathit{bit}_0,\mathit{bit}_1)=\\
& H(\ell) - \frac{|\ell_{00}|}{|\ell|} H(\ell_{00}) - \frac{|\ell_{01}|}{|\ell|} H(\ell_{01})
- \frac{|\ell_{10}|}{|\ell|} H(\ell_{10}) - \frac{|\ell_{11}|}{|\ell|} H(\ell_{11})
\end{align*}
The {\em 2-look-ahead information gain of $\mathit{bit}\in \{1,\ldots,d\}$} is defined as \myspace
\[
\mathit{IG}(\mathit{bit})=\max_{\mathit{bit}_0,\mathit{bit}_1} \mathit{IG}(\mathit{bit},\mathit{bit}_0,\mathit{bit}_1)
\]
We say that $\mathit{bit}\in \{1,\ldots,d\}$ {\em maximizes the 2-look-ahead information gain} if \[
\mathit{bit}\in \argmax \mathit{IG}
\]

In general, we define the \emph{$k$-step weighted entropy} of a node $\ell\subseteq\{0,1\}^d$ with respect to a predicate $\mathit{bit}\in \{1,\ldots,d\}$ by
\begin{align*}
\mathit{WE^k}(\ell,\mathit{bit})=
\min_{\mathit{bit}_0,\mathit{bit}_1} &
\mathit{WE^{k-1}}(\{\vec x\in\ell\mid x_{bit}=0\},\mathit{bit_0})\\+&
\mathit{WE^{k-1}}(\{\vec x\in\ell\mid x_{bit}=1\},\mathit{bit_1})
\end{align*}\myspace\myspace
and
\[\mathit{WE^0}(\ell,\mathit{bit})=|\ell|\cdot H(\ell)\]
Then we say that $\hat{\mathit{bit}}\in \{1,\ldots,d\}$ {\em maximizes the $k$-look-ahead information gain in $\ell$} if \[
\hat{\mathit{bit}}\in \argmax_{bit\in \{1,\ldots,d\}} \left(H(\ell)-\mathit{WE^k}(\ell,\mathit{bit})/|\ell|\right)
=\argmin \mathit{WE^k}(\ell,\cdot)
\]

Note that 1-look-ahead coincides with the choice of split by ID3. For a fixed $k$, if the information gain for
each $i$-look-ahead, $i \leq k$ is zero, we split based on a heuristic on Line~\ref{alg:learn:heur} of Algorithm~\ref{alg:learn}.
This heuristic is detailed on in the following subsection.
Note that Algorithm~\ref{alg:learn} is correct-by-construction since we enforce representation of the
entire input training set. We present a formal correctness proof in Appendix~\ref{app:algcorr}.

\begin{remark}[Properties of look-ahead algorithm]
We now highlight some desirable properties of the look-ahead algorithm.
\begin{compactitem}
\item {\em Incrementality.} First, the algorithm presents an incremental 
approach: computation of the $k$-look-ahead can be done by further refining 
the results of the $(k-1)$-look-ahead analysis due to the recursive nature 
of our definition.
Thus the algorithm can start with $k=2$ and increase $k$ only when required.

\item {\em Entropy-based minimization.} Second, the look-ahead approach
naturally extends the predicate choice of ID3, and thus the entropy-based 
minimization for decision trees is still applicable.

\item {\em Reduction of dimensionality.} 
Finally, Algorithm~\ref{alg:learn} uses the look-ahead method in an incremental fashion, 
thus only considering more complex ``combinations'' when necessary. 
Consequently, we do not produce all these combinations of predicates in advance, 
and avoid the problem of too high dimensionality and only experience local blowups.
\end{compactitem}
In general, $k$-look-ahead clearly requires resources exponential in $k$. However, in our benchmarks, it was typically
sufficient to apply the look-ahead for $k$ equal to two, which is computationally feasible.
\end{remark}
A different look-ahead-based technique was considered in order to dampen the greedy nature of decision tree
construction~\cite{Elomaa2003}, examining the predicates yielding the highest information gains.  
In contrast, our technique retains the greedy approach but focuses on the case where none of the predicates
provides any information gain itself at all and thus  ID3-based techniques fail to advance.
The main goal of our technique is to capture strong dependence between the features of the training set, in
order to solve a different problem than the one treated by~\cite{Elomaa2003}.
Moreover, the look-ahead description in \cite{Elomaa2003} is very informal, which prevents us from implementing
their solution and comparing the two approaches experimentally.

\begin{algorithm}[!t]
	\caption{$k$-look-ahead ID3\label{alg:learn}}
	\begin{algorithmic}[1]
		\Statex \textbf{Inputs:} $\mathit{Train}\subseteq\{0,1\}^{d}$ partitioned into subsets $\mathit{Good}$ and $\mathit{Bad}$.
		\Statex \textbf{Outputs:} A decision tree $\mathcal{T}$ such that $\mathcal{L}(\mathcal{T})\cap\mathit{Train}=\mathit{Good}$.
		\Statex /* train $\mathcal{T}$ on positive set $\mathit{Good}$ and negative set $\mathit{Bad}$ */
		\State $\mathcal{T}\gets(\{\mathit{Train}\},\emptyset,\{\mathit{Train}\mapsto^\theta\good\})$	
		\While{a mixed leaf $\ell$ exists}\label{alg:learn:mixedleaf}
			\If {$\exists \mathit{bit}\in\{1,\ldots,d\}$ with a positive 1-look-ahead information gain}\label{alg:learn:ig}
				\State $\mathit{bit}\gets$ an element of $\{1,\ldots,d\}$ that maximizes
				the 1-look-ahead information gain
				\Statex \Comment maximum information gain is positive
			\Statex $\qquad\vdots$ 
			\ElsIf {$\exists \mathit{bit}\in\{1,\ldots,d\}$ with a positive k-look-ahead information gain} 
			\State $\mathit{bit}\gets$ an element of $\{1,\ldots,d\}$ that maximizes
			the k-look-ahead information gain
			\Statex \Comment maximum $k$-look-ahead information gain is positive	
			\Else
				\State $\mathit{bit}\!\!\gets\!\!\arg\max_{i\in\{1,..,d\}} \max \left\{\!\frac{|\ell[i=0]\cap
				\mathit{Bad}|}{|\ell[i=0]|}\! +\! \frac{|\ell[i=1]\cap \mathit{Good}|}{|\ell[i=1]|}, \! \frac{|\ell[i=0]\cap
				\mathit{Good}|}{|\ell[i=0]|} \! + \! \frac{|\ell[i=1]\cap \mathit{Bad}|}
				{\ell[i=1]|} \right\}$\label{alg:learn:heur}
			\EndIf
			\State split $\ell$ on $\mathit{bit}$ into two leaves $\ell_0$ and $\ell_1$, $\rho(\ell) = \mathit{bit}$\label{alg:learn:rho} 
			\State $\theta(\ell_0)\gets maxclass(\ell_0)$ and $\theta(\ell_1)\gets
			maxclass(\ell_1)$\label{alg:learn:theta} 
		\EndWhile	
		\State\Return $\mathcal{T}$
	\end{algorithmic}
\end{algorithm}

\subsection{Heuristics}\label{subs:heuristics}

\subsubsection{Statistical split-decision.}
The look-ahead mentioned above provides a very systematic principle on
how to resolve splitting decisions. However, the computation can be demanding in
terms of computational resources. Therefore we present a very simple statistical
heuristic that gives us one more option to decide a split. The precise formula
is $\mathit{bit}=$
\[
\argmax_{i \in \{1,..,d\}} \max \!
		\bigg\{ \! \frac{|\ell[i=0]\cap \mathit{Bad}|}{|\ell[i=0]|}  +  \frac{|\ell[i=1]\cap \mathit{Good}|}{|\ell[i=1]|},
				\! \frac{|\ell[i=0]\cap \mathit{Good}|}{|\ell[i=0]|} +  \frac{|\ell[i=1]\cap \mathit{Bad}|}{\ell[i=1]|} \! \bigg\}
\]
Intuitively, we choose a $\mathit{bit}$ that maximizes the portion
of good samples in one subset and the portion of bad samples in the other subset,
which mimics the entropy-based method, and at the same time is very fast to compute.
One can consider using this heuristic exclusively every time the basic ID3-based splitting
technique fails. However, in our experiments, using 2-look-ahead and then (once needed) proceeding
with the heuristic yields better results, and is still computationally undemanding.

\subsubsection{Chain disjunction.}
The entropy-based approach favors the splits where one of the branches contains a completely
resolved data set ($\ell_* \!\subseteq\! \mathit{Good}$ or $\ell_* \!\subseteq\! \mathit{Bad}$), as
this provides notable information gain. Therefore, as the algorithm proceeds, it often happens that
at some point multiple splits provide a resolved data set in one of the branches. We consider
a heuristic that \emph{chains} all such splits together and computes the information gain of the
resulting disjunction. More specifically, when considering each
$\mathit{bit}$ as a split candidate (line~\ref{alg:learn:ig}~of~Algorithm~\ref{alg:learn}),
we also consider (a)~the disjunction of all bits that contain a subset of $\mathit{Good}$
in either of the branches, and (b)~the disjunction of bits containing a subset of $\mathit{Bad}$ in a branch.
Then we choose the candidate that maximizes the information gain.
These two extra checks are very fast to compute, and can improve succinctness and readability of
the decision trees substantially, while maintaining the fact that a decision tree fits its training set exactly.
Appendix~\ref{app:bitshifter} 
provides two examples where the decision tree obtained without
this heuristic is presented, and then the decision tree obtained when using the heuristic is presented.

\section{Experimental Results}\label{sec:exper}

\newcommand{\tilda}{{\sim}}

In our experiments we use two sources of problems reducible to the representation of memoryless strategies in I/O
games with binary variables: AIGER specifications~\cite{AIGER} and LTL specifications~\cite{Pnueli77}.
Given a game, we use an explicit solver to obtain a strategy in the form of a list of played and non-played actions for each state,
which can be directly used as a training set.
Throughout our experiments, we compare succinctness of representation (expressed as the number of inner nodes) using decision trees and BDDs.

We implemented our method in the programming language Java. We used the external library \textsc{CuDD}~\cite{Somenzi15} for the manipulation of BDDs.
We used the Algorithm~\ref{alg:learn} with $k=2$ to compute the decision trees. We obtained all the results on a single machine
with Intel(R) Core(TM) i5-6200U CPU (2.40 GHz) with the heap size limited to 8 GB.

\subsection{AIGER specifications}\label{subs:experaiger}

SYNTCOMP \cite{DBLP:journals/corr/JacobsBBKPRRSST16} is the most important competition of synthesis tools, running yearly since~2014.
Most of the benchmarks have the form of AIGER specifications \cite{AIGER}, describing safety specifications using circuits with input,
output, and latch variables. This reduces directly to the I/O games with variables since the latches describe the current configuration
of the circuit, corresponding to the state variables of the game.
Since the objectives here are safety/reachability, the winning strategies can be computed and guaranteed to be memoryless.

We consider two benchmarks: scheduling of washing cycles in a washing system and a simple bit shifter model
(the latter presented only in Appendix~\ref{app:bitshifter} 
due to space constraints),
introduced in SYNTCOMP 2015~\cite{DBLP:journals/corr/JacobsBBKPRRSST16} and SYNTCOMP 2014, respectively.

\subsubsection{Scheduling of Washing Cycles.}

The goal is to design a centralized controller for a washing system, composed of several tanks running in
parallel~\cite{DBLP:journals/corr/JacobsBBKPRRSST16}. The model of the system is parametrized by the number of tanks,
the maximum allowed reaction delay before filling a tank with water, the delay after which the tank has to be emptied
again, and the number of tanks that share a water pipe. The controller should satisfy a safety objective, that is,
avoid reaching an error state, which means that the objective of the other player is reachability.
In total, we obtain 406 graph games with safety/reachability objectives.
In 394 cases we represent a winning strategy of the safety player, in the remaining 12 cases a winning strategy of the
reachability player. The number of states of the graph games ranges from 30 to 43203, the size of training example
sets ranges from 40 to 3359232.

\setlength{\abovecaptionskip}{0pt}
\begin{figure}
\begin{minipage}{.5\textwidth}
	\begin{center}
	\includegraphics[scale=0.44]{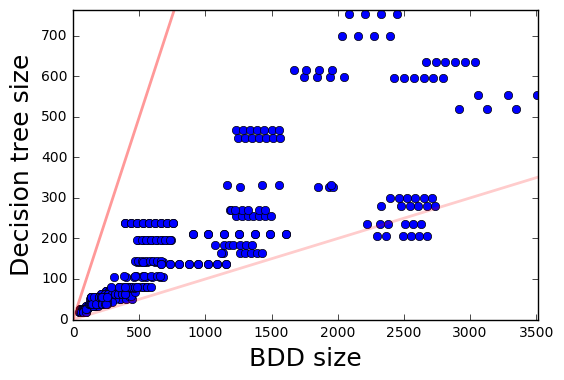}
	\end{center}
\end{minipage}
\begin{minipage}{.5\textwidth}
	\begin{center}
	\includegraphics[scale=0.44]{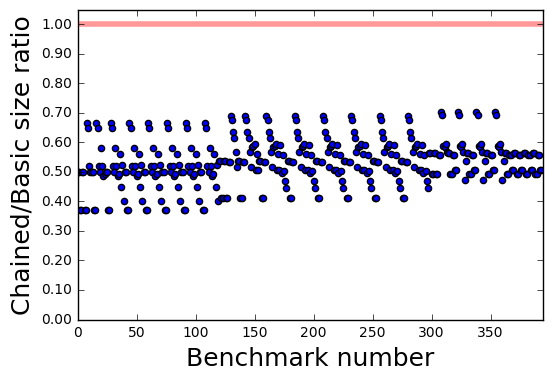}
	\end{center}
\end{minipage}
\caption{Washing cycles -- safety}
\label{fig:washsafe}
\end{figure}

The left plot in Fig.~\ref{fig:washsafe} displays the size of our decision tree representation of the controller winning safety
strategies versus the size of their BDD representations. The decision tree is smaller than the corresponding BDD in all 394 cases.
The arithmetic average ratio of decision tree size and BDD size is $\tilda24\%$, the geometric average is $\tilda22\%$, and the harmonic
average is $\tilda21\%$.

In these experiments, we obtain the BDD representation as follows: we consider 1000 randomly chosen variable orderings
and for each construct a corresponding BDD, in the end we consider the BDD with the minimal size. As a different set of
experiments, we compare against BDDs obtained using several algorithms for variable reordering, namely,
Sift~\cite{Rudell93}, Window4~\cite{Fujita91}, simulated-annealing-based algorithm~\cite{Bollig95}, and a genetic
algorithm~\cite{Drechsler95}. The results with these algorithms are very similar and provided in Appendix~\ref{app:wash}. 
Furthermore, the information about execution time is also provided in Appendix~\ref{app:wash}. 

Moreover, in the experiments described above, we do not use the chain heuristic described in Section~\ref{subs:heuristics}, in order
to provide a fair comparison of decision trees and BDDs. The right plot in Fig.~\ref{fig:washsafe} displays the difference
in decision tree size once the chain heuristic is enabled. Each dot represents the ratio of decision tree size with
and without it.

The decision trees also allow us to get some insight into the winning strategies. Namely, for a fixed number
of water tanks and a fixed empty delay, we obtain a solution that is affected by different values of the fill delay in a minimal
way, and is easily generalizable for all the values of the parameter. This fact becomes more apparent once the chain heuristic
described in Section~\ref{subs:heuristics} is enabled. This phenomenon is not present in the case of BDDs as they differ significantly, even
in size, for different values of the parameter (see Table~\ref{fig:washsafesnippet} in Appendix~\ref{app:wash}). 
For two tanks and empty delay of one, the solution is small enough to be humanly readable
and understandable, see Fig.~\ref{fig:washsafeuniv} (where the fill delay is set to $7$).
Additional examples of the parametric solutions can be found in Appendix~\ref{app:wash}. 
This example suggests that decision tree representation might be useful in solving parametrized synthesis (and verification) problems.

\begin{figure}
	\begin{center}
	\includegraphics[width=9cm, height=8cm]{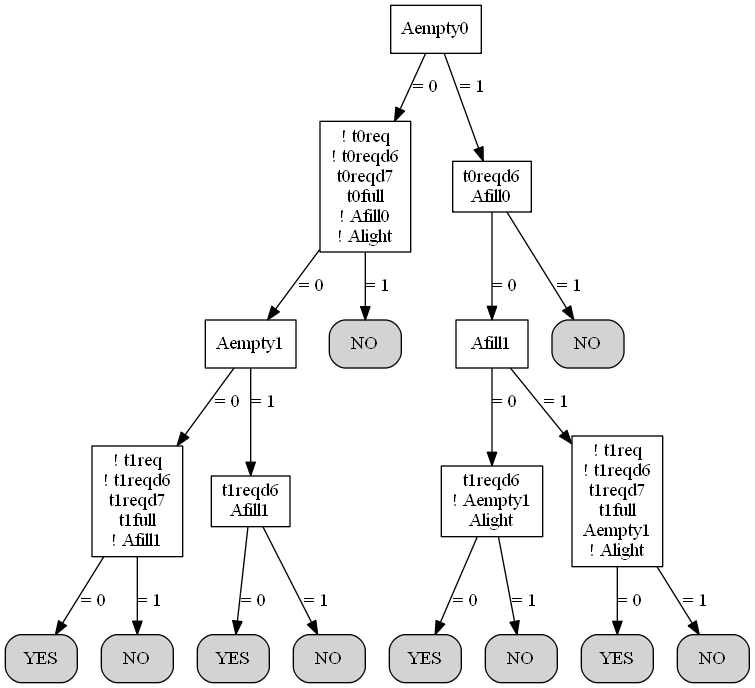}
	\end{center}
\caption{A solution for two tanks and empty delay of one, illustration for fill delay of $7$. Solution for other values $p$ are the same except for replacing values $p$ and $p-1$ for $7$ and $6$, respectively. Thus a parametric solution could be obtained by a simple syntactic analysis of the difference of any two instance solutions.}
\label{fig:washsafeuniv}
\end{figure}

\setlength{\abovecaptionskip}{2pt}
\begin{figure}
\begin{minipage}{.62\textwidth}
\begin{tabular}{ l | l | l | l | l || r | r | r}
\hline
Name & $|S|$ & $|I|$ & $|O|$ & $|\mathit{Train}|$ & $|BDD|$ & $|DT|$ & $|DT\!+\!|$ \\
\hline
wash\_3\_1\_1\_3      &      102  &   3  &   7   &       40    &      45  &    3   &  1    \\
wash\_4\_1\_1\_3      &      466  &   4  &   9   &      144    &      76  &    4   &  1    \\
wash\_4\_1\_1\_4      &      346  &   4  &   9   &       96    &      78  &    4   &  1    \\
wash\_4\_2\_1\_4      &      958  &   4  &   9   &      432    &     157  &    4   &  1    \\
wash\_4\_2\_2\_4      &     3310  &   4  &   9   &      432    &     301  &    4   &  1    \\
wash\_5\_1\_1\_3      &     1862  &   5  &  11   &      416    &     127  &    5   &  1    \\
wash\_5\_1\_1\_4      &     1630  &   5  &  11   &      352    &     121  &    5   &  1    \\
wash\_5\_2\_1\_4      &     5365  &   5  &  11   &     2368    &     255  &    5   &  1    \\
wash\_5\_2\_2\_4      &    27919  &   5  &  11   &     2368    &     554  &    5   &  1    \\
wash\_6\_1\_1\_3      &     6962  &   6  &  13   &     1088    &     193  &    6   &  1    \\
wash\_6\_1\_1\_4      &     6622  &   6  &  13   &     1024    &     172  &    6   &  1    \\
wash\_6\_2\_1\_4      &    27412  &   6  &  13   &    10432    &     419  &    6   &  1    \\
\hline  
\end{tabular}
\end{minipage}
\begin{minipage}{.3\textwidth}
	\includegraphics[width=5cm, height=4cm]{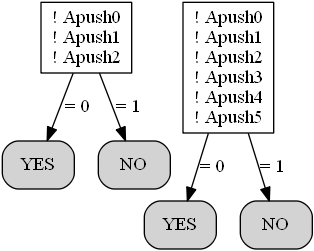} 
\end{minipage}
\caption{Washing cycles -- reachability}
\label{fig:washreach}
\end{figure}
\setlength{\abovecaptionskip}{0pt}

The table in Fig.~\ref{fig:washreach} summarizes the results for the cases where the controller cannot be synthesized and
we synthesize a counterexample winning reachability strategy of the environment. The benchmark parameters specify the total number of tanks,
the fill delay, the empty delay, and the number of tanks sharing a pipe, respectively. In all of these cases, the size of the
decision tree is substantially smaller compared to its BDD counterpart. The decision trees also provide some structural
insight that may easily be used in debugging. Namely, the trees have a simple repeating structure where the number of repetitions
depends just on the number of tanks. This is even easier to see once the chain heuristic of Section~\ref{subs:heuristics} is used.
Fig.~\ref{fig:washreach} shows the tree solution for the case of three and six tanks, respectively.
The structural phenomenon is not apparent from the BDDs at all.

\subsection{Random LTL}\label{subs:experrandomltl}

In reactive synthesis, the objectives are often specified as LTL (linear-time temporal logic) formulae over input/output letters.
In our experiments, we use formulae randomly generated using SPOT~\cite{spot}
\footnote{First, we run randltl from the Spot tool-set \texttt{randltl -n10000 5 --tree-size=20..25 –seed=0 --simplify=3 -p --ltl-priorities ’ap=3,
false=1,true=1,not=1,F=1,G=1,X=1,equiv=1,implies=1,xor=0,R=0,U=1,
W=0,M=0,and=1,or=1’ | ltlfilt --unabbreviate="eiMRWˆ"}
to obtain the formulae. Then we run Rabinizer to obtain the respective automata and we retain those with at least 100 states.}.
LTL formulae can be translated into deterministic parity automata; for this translation we use the tool Rabinizer~\cite{rabinizer}. 
Finally, given a parity automaton, we consider various partitions of the atomic propositions into input/output letters, which
gives rise to graph games with parity objectives. See Appendix~\ref{app:parity} 
for more details on the translation. We retain all formulae that result in games with at most three priorities.

Consequently, we use two ways of encoding states of the graph games as binary vectors. First, {\em naive encoding}, allowed by the
fact that the output of tools such as \cite{spot,rabinizer} in HOA format \cite{DBLP:conf/cav/BabiakBDKKM0S15} always assigns an
id to each state. As this id is an integer, we may use its binary encoding.
Second, we use a more sophisticated {\em Rabinizer encoding} obtained by using internal structure of states produced by
Rabinizer~\cite{rabinizer}. Here the states are of the form ``formula, set of formulae, permutation, priority''.
We propose a very simple, yet efficient procedure of encoding the state structure information into bitvectors.
Although the resulting bitvectors are longer than in the naive encoding, some structural information of the game is preserved,
which can be utilized by decision trees to provide a more succinct representation.
BDDs perform a lot better on the naive encoding than on the Rabinizer encoding, since they are unable to exploit the preserved
state information. As a result, we consider the naive encoding with BDDs and both, the naive
and the Rabinizer encodings, with decision trees.

We consider 976 examples where the goal of the player, whose strategy is being represented,
is that the least priority occurring an infinite number of times is odd. 

\setlength{\abovecaptionskip}{-10pt}

\begin{figure}
	\begin{minipage}{.5\textwidth}
		\begin{center}
        \includegraphics[scale=0.45]{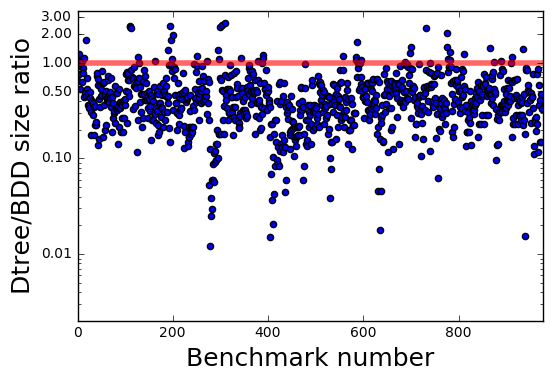}
		\end{center}
		\caption{BDDs vs DTrees}
		\label{fig:random1}
	\end{minipage}
	\begin{minipage}{.5\textwidth}
		\begin{center}
        \includegraphics[scale=0.45]{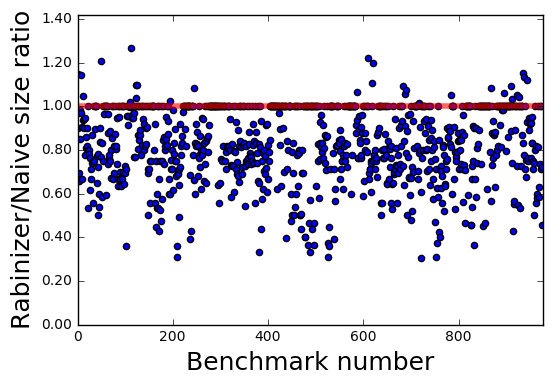}
		\end{center}
		\caption{DTrees improvement with Rabinizer enc.}
		\label{fig:random2}
	\end{minipage}
\end{figure}

\setlength{\abovecaptionskip}{4pt}

Fig.~\ref{fig:random1} plots the size ratios when we compare BDDs and decision trees (note that the $y$-axis scales logarithmically).
For each case, we consider 1000 random variable orderings and choose the BDD that is minimal in size, and after that we construct
a decision tree (without the chain heuristic of Section~\ref{subs:heuristics}). For BDDs, we also consider all the ordering algorithms mentioned
in the previous set of experiments, however, they provide no improvement compared to the random orderings.

In 925 out of 976 cases, the resulting decision tree is smaller than the corresponding BDD (in 3 cases they are of a same size and
in 48 cases the BDD is smaller). The arithmetic average ratio of decision tree size and BDD size is $\tilda46\%$, the geometric
average is $\tilda38\%$, and the harmonic average is $\tilda28\%$.

Fig.~\ref{fig:random2} demonstrates how decision tree representation improves once the features of the
game-structural information can be utilized. Each dot corresponds to a ratio of the decision tree size once the Rabinizer encoding is used,
and once the naive encoding is used. In 638 cases the Rabinizer encoding is superior, in 309 cases there is no
difference, and in 29 cases the naive encoding is superior. All three types of the average ratio are around $80\%$.
In Appendix~\ref{app:randltl} 
we present the further improvement of decision trees once we use the chain heuristic of Section~\ref{subs:heuristics}.

\section{Conclusion}\label{sec:concl}
In this work we propose decision trees for strategy representation in graph 
games.
While decision trees have been used in probabilistic settings where errors are
allowed and overfitting of data is avoided, for graph games, strategies
must be entirely represented without errors. 
Hence optimization techniques for existing decision-tree solvers do not apply,
and we develop new techniques and present experimental results to demonstrate
the effectiveness of our approach.
Moreover, decision trees have several other advantages:
First, in decision trees the nodes represent predicates, and in richer domains, 
e.g., where variables represent integers, the internal nodes of the tree can 
represent predicates in the corresponding domain, e.g., comparison between the 
integer variables and a constant. 
Hence richer domains can be directly represented as decision trees without 
conversion to bitvectors as required by BDDs. 
However, we restricted ourselves to the boolean domain to show that even in 
such domains that BDDs are designed for the decision trees improve over BDDs.
Second, as illustrated in our examples, decision trees can often provide similar and
scalable solution when some parameters vary. 
This is quite attractive in reactive synthesis where certain parameters vary, 
however they affect the strategy in a minimal way. 
Our examples show decision trees exploit this much better than BDDs, and can 
be useful in parametrized synthesis.
Our work opens up many interesting directions of future work. 
For instance, richer versions of decision trees that are still well-readable could be used instead,
such as decision trees with more complex expressions in leaves \cite{DBLP:conf/tacas/NeiderSM16}.
The applications of decision trees in other applications related to 
reactive synthesis is an interesting direction of future work.
Another interesting direction is the application of the look-ahead technique 
in the probabilistic settings.


\subsubsection*{Data Availability Statement and Acknowledgments.}This work has been partially supported by
the Czech Science Foundation, Grant No.~P202/12/G061, 
Vienna Science and Technology Fund (WWTF) Project ICT15-003, 
Austrian Science Fund (FWF) NFN Grant No.~S11407-N23 (RiSE/SHiNE), 
ERC Starting grant (279307: Graph Games), 
DFG Grant No KR 4890/2-1 (SUV: Statistical Unbounded Verification), 
TUM IGSSE Grant 10.06 (PARSEC) 
and EU Horizon 2020 research and innovation programme under the Marie Sk\l{}odowska-Curie Grant No.~665385. 
We thank Fabio Somenzi for detailed information about variable reordering in BDDs.
The source code and binary files used to obtain the results presented in this paper
are available in the figshare repository: https://doi.org/10.6084/m9.figshare.5923915~\cite{Artifact}.

{
\bibliographystyle{abbrv}
\bibliography{diss,related-work}
}

\vspace{6mm}
\section*{Appendix}
\appendix

\setlength\floatsep{2pt}
\setlength\textfloatsep{2pt}
\setlength\intextsep{2pt}
\setlength\abovecaptionskip{0pt}
\setlength\belowcaptionskip{0pt}

\section{Artifact Description}

We provide instructions to replicate the experimental results presented in this paper, using our artifact that
is openly available at~\cite{Artifact}. All the results can be obtained
with the heap size limited to 8 GB.

\vspace{3mm}
\smallskip\noindent{\bf Results for Scheduling of Washing Cycles (Section~\ref{subs:experaiger}).}
Running this batch takes roughly 30 hours and generates 7.1GB of training data. Note that
we did not include around 30 most resource-demanding benchmarks of this batch in the artifact.
(i) in folder art, execute ./run.sh wTOTAL,
(ii) observe the results at art/results/reports/reprWash\{2,3,4,reach\}.txt,
(iii) in folder art/results, execute python plotsWash.py and observe the plots that correspond to Figure~\ref{fig:washsafe}.
Alternatively, to run a subset of this batch that takes only 30 minutes to run and generates only 265MB of training data,
in (i) execute ./run.sh wPART. To additionaly generate dot representation of DTs/BDDs, in (i) execute
either ./run.sh wTOTALdot or ./run.sh wPARTdot.

\vspace{3mm}
\smallskip\noindent{\bf Results for Scheduling of Washing Cycles BDD reordering (Appendix~\ref{app:wash}).} 
Running this batch takes roughly 30 minutes.
(i) make sure you have the training data obtained by running the batch above,
(ii) in folder art/results, execute ./runBDDreorder.sh,
(iii) observe the results at art/results/reports/BDDreorder.txt.

\vspace{3mm}
\smallskip\noindent{\bf Results for Random LTL (Section~\ref{subs:experrandomltl}).}
Running this batch takes roughly 2 hours and generates 84MB of training data.
(i) in folder art, execute ./run.sh rTOTAL,
(ii) observe the results at art/results/reports/reprRandomLTL\{naive,encoded\}.txt,
(iii) in folder art/results, execute python plotsRandomLTL.py and observe the plots
that correspond to Figure~\ref{fig:random1} and Figure~\ref{fig:random2}.

\vspace{3mm}
\smallskip\noindent{\bf Results for Bit Shifter (Appendix~\ref{app:bitshifter}).} 
Running this experiment batch takes roughly 5 minutes. Note that we did not include
two benchmarks in the artifact since they take considerable execution time.
(i) in folder art, execute ./run.sh aTOTAL,
(ii) observe the results at art/results/reports/reprAiger.txt.

\clearpage

\section{Correctness of Algorithm $k$-look-ahead ID3}\label{app:algcorr}

\begin{theorem}
Let $G$ be an I/O game with binary variables, and let $\sigma \colon S^R_*\rightarrow A_*$ be a memoryless
strategy that defines a training set $\mathit{Train}$. Algorithm~\ref{alg:learn} with input $\mathit{Train}$ outputs
a decision tree $\mathcal{T}$ such that $\mathcal{L}(\mathcal{T})\cap\mathit{Train}=\mathit{Good}$,
which means that for all $s\in S^R_*$ we have that $\tuple{s,a}\in \mathcal{L}(\mathcal{T})$ iff $\sigma(s)=a$.
Thus $\mathcal{T}$ represents the strategy $\sigma$.

\end{theorem}
\begin{proof}
Recall that a strategy $\sigma \colon S^R_*\rightarrow A_*$ defines:
\begin{itemize}
\item $\mathit{Good} = \{\tuple{s,\sigma(s)}\in S^R_*\times A_*\}$
\item $\mathit{Bad} = \{\tuple{s,a} \in S^R_*\times A_* \mid a\not = \sigma(s)\}$
\item $\mathit{Train} = \mathit{Good} \uplus \mathit{Bad}$ ($\uplus$ denotes a disjoint union)
\end{itemize}
Since we consider I/O games with binary variables, states and actions are labeled by bitvectors,
so $\mathit{Train} \subseteq \{0,1\}^{d}$, where $d$ is the number of features for states plus the
number of features for actions.
Further recall that given a leaf $\ell$, we define $\mathit{maxclass}(\ell)$ as $\good$ if
$|\ell\cap\mathit{Good}|\geq |\ell\cap\mathit{Bad}|$, and $\bad$ otherwise.
Also, a leaf is mixed if it has a non-empty intersection with both $\mathit{Good}$ and $\mathit{Bad}$.
Finally recall that given a decision tree $\mathcal{T}=(T,\rho,\theta)$,
$\rho$ assigns to every inner node a number from $\{1,\ldots,d\}$,
and $\theta$ assigns to every leaf a value $\good$ or $\bad$.

\paragraph{Partial correctness.}
Consider the algorithm with input $\mathit{Train}$, and let $\mathcal{T}$ be the output decision tree. Consider
arbitrary $\tuple{s,a} \in S^R_*\times A_*$, note that it belongs to $\mathit{Train}$. Consider the leaf $\ell$ corresponding
to $\tuple{s,a}$ in $\mathcal{T}$, i.e., $\tuple{s,a} \in \ell$. Decision tree $\mathcal{T}$ has no mixed leaves, since
otherwise the main while-loop (line~\ref{alg:learn:mixedleaf}) and the algorithm would not have terminated.
Therefore $\ell$ is not mixed and thus we have that either (i) $\ell \cap \mathit{Train} \subseteq \mathit{Good}$,
implying $\mathit{maxclass}(\ell) = \good$, or (ii) $\ell \cap \mathit{Train} \subseteq \mathit{Bad}$,
implying $\mathit{maxclass}(\ell) = \bad$. Additionaly, $\theta(\ell) = \mathit{maxclass}(\ell)$, which was set
at line~\ref{alg:learn:theta} during the iteration of the main while-loop that processed the parent of $\ell$.
Since $\tuple{s,a} \in \ell \cap \mathit{Train}$, we obtain $\theta(\ell)=\good$ iff $\tuple{s,a}\in \mathit{Good}$.
Finally, since  $\tuple{s,a}\in \mathcal{L}(\mathcal{T})$ iff $\theta(\ell)=\good$ and $\tuple{s,a}\in \mathit{Good}$
iff $\sigma(s)=a$, we obtain $\tuple{s,a}\in \mathcal{L}(\mathcal{T})$ iff $\sigma(s)=a$.

\paragraph{Total correctness.}
We maintain an invariant such that the length of an arbitrary path in $\mathcal{T}$ is at most $d$, i.e., the number
of features in $\mathit{Train}$. We prove this by showing that for every feature $x\in\{1,\ldots,d\}$, for every path
in $\mathcal{T}$, at most one inner node $\bar{n}$ of the path has $\rho(\bar{n})=x$.

Consider a path with a mixed leaf $\ell$, let $X \subset \{1,..,d\}$ be the set of features appearing in this path.
All elements of $\ell \cap \mathit{Train}$ coincide in the values of the features of $X$. Additionally, from the
definition of a mixed leaf it follows that $X$ is indeed a strict subset of $\{1,..,d\}$ and that there exists an
element of $\ell \cap \mathit{Good}$ and an element of $\ell \cap \mathit{Bad}$, let $y \notin X$ be an
arbitrary feature where these two elements differ. Consider the smallest $k$ such that the
maximum $k$-look-ahead information gain is positive (such $k$ exists and its value is bounded by $d-|X|$).
For every $x \in X$, its $k$-look-ahead information gain is zero, since (i) its information gain is zero, and
(ii) in case $k>1$, its $k$-look-ahead information gain is bounded by the maximum $(k\!-\!1)$-look-ahead
information gain, which is zero. Therefore the feature maximizing $k$-look-ahead information gain
does not belong to $X$. When the heuristic at line~\ref{alg:learn:heur} is computed, the value of the
computation formula for feature $y$ is positive, whereas for every $x \in X$, we obtain undefined
terms $\frac{0}{0}$ in the computation formula, so we explicitly define the value as 0. Therefore the
feature maximizing the value of the formula does not belong to $X$. Finally, at line~\ref{alg:learn:rho}
of the algorithm, we set $\rho(\ell) = \mathit{bit}$, and by the arguments presented above, $\mathit{bit} \notin X$.

Since every iteration of the while-loop adds two vertices to the decision tree, by the above invariant we have
that the algorithm terminates. This together with partial correctness gives us total correctness.\qed

\end{proof}

\section{Details of Section~\ref{subs:experaiger}: Scheduling of Washing Cycles}\label{app:wash}

\smallskip\noindent{\bf Execution time.}
The average time spent on constructing a decision tree for a given benchmark is around $32$ seconds.
For BDDs, the construction algorithm uses the optimized \textsc{CuDD}~\cite{Somenzi15} library,
and constructs faster. Therefore, we consider $1000$ randomly chosen variable orderings and in the
end retain the smallest BDD. This way, we provide a lot higher time budget for the BDD construction,
the average time spent on one example is around $327$ seconds.
Finally, the average time spent on one example for decision trees with the chain heuristic
(see Section~\ref{subs:heuristics}) is around $23$ seconds. Note that this is shorter than when
the heuristic is turned off. This shows the heuristic incurs minimal computational overhead,
and on the other hand saves resources when performing essentially multiple splits at once.

\vspace{2mm}
\smallskip\noindent{\bf Reordering algorithms.}
We compare decision trees (obtained without the chain heuristic of Section~\ref{subs:heuristics}), and
BDDs obtained as follows. For each example, we consider four BDDs obtained using the following
reordering algorithms: Sift~\cite{Rudell93}, Window4~\cite{Fujita91}, simulated-annealing-based
algorithm~\cite{Bollig95}, and a genetic algorithm~\cite{Drechsler95}.  Then, we retain the
smallest BDD.

In 332 out of 394 cases, none of the algorithms provide any improvement compared to the 
default variable ordering. This suggests the default natural ordering is already quite solid for the
BDD representation. Fig.~\ref{fig:washsafereorder} plots the results of the comparison, the red
dots correspond to the cases where reordering algorithms provide improvement.
The decision tree is smaller in 386 cases, and the BDD is smaller in 8 cases.
The arithmetic average ratio of decision tree size and BDD size is $\tilda28\%$, the geometric
average is $\tilda25\%$, and the harmonic average is $\tilda23\%$.

\begin{figure}
	\begin{center}
		\includegraphics[scale=0.52]{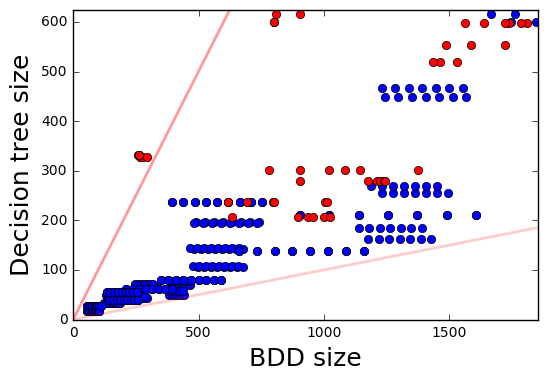}
	\end{center}
	\vspace{-4mm}
	\caption{Washing cycles -- safety; BDDs with reordering algorithms}
	\label{fig:washsafereorder}
\end{figure}

\smallskip\noindent{\bf Parametric solutions -- comparison with BDDs.}
In Table~\ref{fig:washsafesnippet} we provide a snippet of the table containing detailed information about the
results of the experiments for safety. The benchmark parameters specify the total number of tanks,
the fill delay, the empty delay, and the number of tanks sharing a pipe, respectively.
The snippet shows how BDDs differ in size for the cases where decision
trees provide a solution that differs minimally and is easily generalizable for all the cases.

\vspace{2mm}
\smallskip\noindent{\bf Parametric solutions -- more examples.}
For two tanks and empty delay of one, in Section~\ref{subs:experaiger} we provide one illustration
of the generalizable solution, for fill delay of $7$. Fig.~\ref{fig:washsafeuniv4} provides an illustration for fill delay
$4$, in order to show how some labels change when the parameter value is changes (note that the
structure of the decision tree remains the same). Finally, Fig.~\ref{fig:washsafeunivp} presents the parametric solution
for all the values of the fill delay, which could be be easily obtained by a syntactic analysis of the
difference of any two instance solutions.

\begin{table}[t]
	\begin{center}
		\begin{tabular}{ l | l | l | l | l || r | r | r}
			\hline
			Name & $|S|$ & $|I|$ & $|O|$ & $|\mathit{Train}|$ & $|BDD|$ & $|DT|$ & $|DT\!+\!|$ \\
			\hline
wash\_2\_1\_1\_1     &       35  &   2  &   5     &        800       &      50   &   18    &     9    \\
wash\_2\_2\_1\_1     &       53  &   2  &   5     &        800       &      54   &   18    &     9    \\
wash\_2\_3\_1\_1     &       75  &   2  &   5     &        800       &      57   &   18    &     9    \\
wash\_2\_4\_1\_1     &      101  &   2  &   5     &        800       &      64   &   18    &     9    \\
wash\_2\_5\_1\_1     &      131  &   2  &   5     &        800       &      74   &   18    &     9    \\
wash\_2\_6\_1\_1     &      165  &   2  &   5     &        800       &      77   &   18    &     9    \\
wash\_2\_7\_1\_1     &      203  &   2  &   5     &        800       &      84   &   18    &     9    \\
wash\_2\_8\_1\_1     &      245  &   2  &   5     &        800       &      85   &   18    &     9    \\
wash\_2\_9\_1\_1     &      291  &   2  &   5     &        800       &      75   &   18    &     9    \\ \hline
wash\_2\_2\_2\_1     &      118  &   2  &   5     &       2592       &      94   &   33    &    22    \\
wash\_2\_3\_2\_1     &      150  &   2  &   5     &       2592       &     121   &   33    &    22    \\
wash\_2\_4\_2\_1     &      186  &   2  &   5     &       2592       &     113   &   33    &    22    \\
wash\_2\_5\_2\_1     &      226  &   2  &   5     &       2592       &     154   &   33    &    22    \\
wash\_2\_6\_2\_1     &      270  &   2  &   5     &       2592       &     138   &   33    &    22    \\
wash\_2\_7\_2\_1     &      318  &   2  &   5     &       2592       &     165   &   33    &    22    \\
wash\_2\_8\_2\_1     &      370  &   2  &   5     &       2592       &     126   &   33    &    22    \\
wash\_2\_9\_2\_1     &      426  &   2  &   5     &       2592       &     185   &   33    &    22    \\ \hline
wash\_3\_1\_1\_1     &      153  &   3  &   7     &      16000       &     139   &   39    &    21    \\
wash\_3\_2\_1\_1     &      281  &   3  &   7     &      16000       &     142   &   39    &    21    \\
wash\_3\_3\_1\_1     &      469  &   3  &   7     &      16000       &     132   &   39    &    21    \\
wash\_3\_4\_1\_1     &      729  &   3  &   7     &      16000       &     181   &   39    &    21    \\
wash\_3\_5\_1\_1     &     1073  &   3  &   7     &      16000       &     185   &   39    &    21    \\
wash\_3\_6\_1\_1     &     1513  &   3  &   7     &      16000       &     215   &   39    &    21    \\
wash\_3\_7\_1\_1     &     2061  &   3  &   7     &      16000       &     217   &   39    &    21    \\
wash\_3\_8\_1\_1     &     2729  &   3  &   7     &      16000       &     217   &   39    &    21    \\
wash\_3\_9\_1\_1     &     3529  &   3  &   7     &      16000       &     253   &   39    &    21    \\

			\hline  
		\end{tabular}
	\end{center}
	\caption{Washing cycles -- safety; snippet of the results}
	\label{fig:washsafesnippet}
\end{table}

We have observed multiple cases of the parametric solution phenomenon.
We present one more example for the case of three tanks and empty delay of one,
in Fig.~\ref{fig:washsafeuniv3tanks} (where the fill delay is set to $7$).

\section{Details of Section~\ref{subs:experaiger}: Bit Shifter}\label{app:bitshifter}

The specification for a bit shifter circuit is one of the toy example benchmarks for SYNTCOMP. The benchmark set
is parametrized by the length of the input bit array.

\vspace{2mm}
\begin{table}
	\begin{center}
		\begin{tabular}{ l | l | l | l | l || r | r | r}
			\hline
			Name & $|S|$ & $|I|$ & $|O|$ & $|\mathit{Train}|$ & $|BDD|$ & $|DT|$ & $|DT\!+\!|$ \\
			\hline
			bs16n & 305 & 4 & 1 & 64 		& 39 & 11 & 3\\
			bs32n & 1121 & 5 & 1 & 128 		& 72 & 13 & 3\\
			bs64n & 4289 & 6 & 1 & 256 		& 137 & 15 & 3\\
			bs128n & 16769 & 7 & 1 & 512 	& 266 & 17 & 3\\
			bs256n & 66305 & 8 & 1 & 1024 	& 523 & 19 & 3\\
			bs512n & 263681 & 9 & 1 & 2048 	& 1036 & 21 & 3\\
			\hline  
		\end{tabular}
	\end{center}
	\caption{Bit shifter}
	\label{fig:bsres}
\end{table}

Fig.~\ref{fig:bsres} summarizes the results. The decision trees are smaller in each case and the difference increases
with the increasing parameter value for the benchmark.


\begin{figure}
	\vspace{-4mm}
	\begin{minipage}{.9\textwidth}
		\begin{center}
        		\includegraphics[width=12cm, height=10cm]{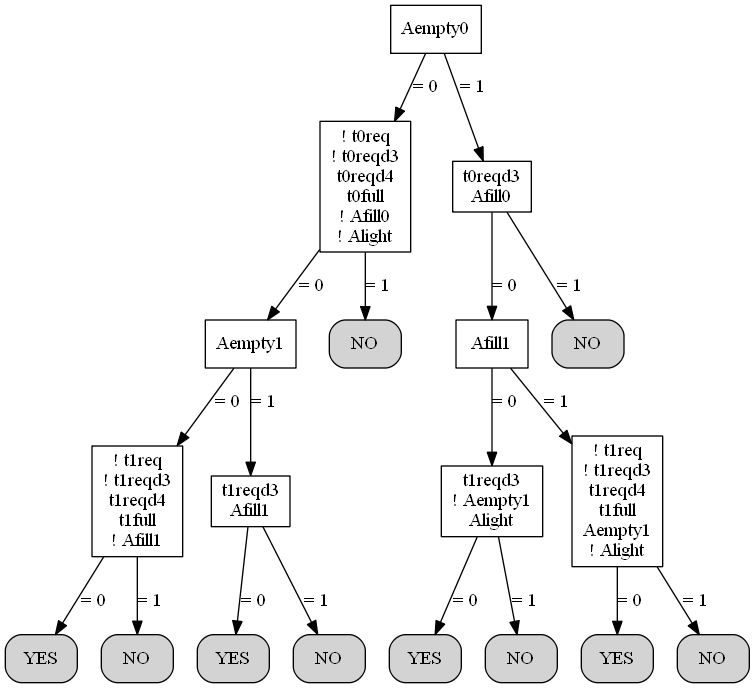}
		\end{center}
		\vspace{-4mm}
		\caption{A solution for two tanks and empty delay of one, illustration for fill delay of $4$.}
		\vspace{1mm}
		\label{fig:washsafeuniv4}
	\end{minipage}
	\begin{minipage}{.9\textwidth}
		\begin{center}
        		\includegraphics[width=12cm, height=10cm]{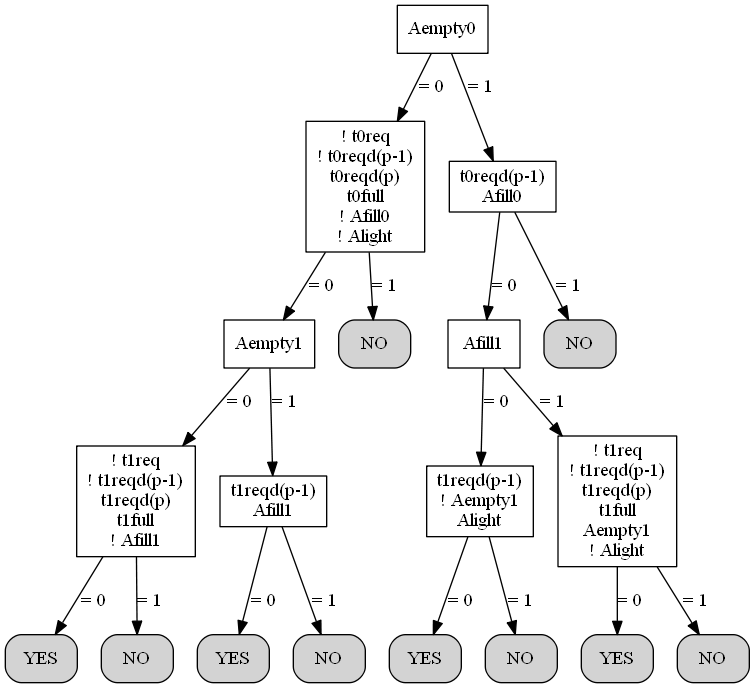}
		\end{center}
		\vspace{-4mm}
		\caption{A parametric solution for two tanks and empty delay of one, fill delay parametrized by $p$.}
		\label{fig:washsafeunivp}
	\end{minipage}
\end{figure}

\begin{figure}[t]
	\begin{center}
	\includegraphics[scale=0.25]{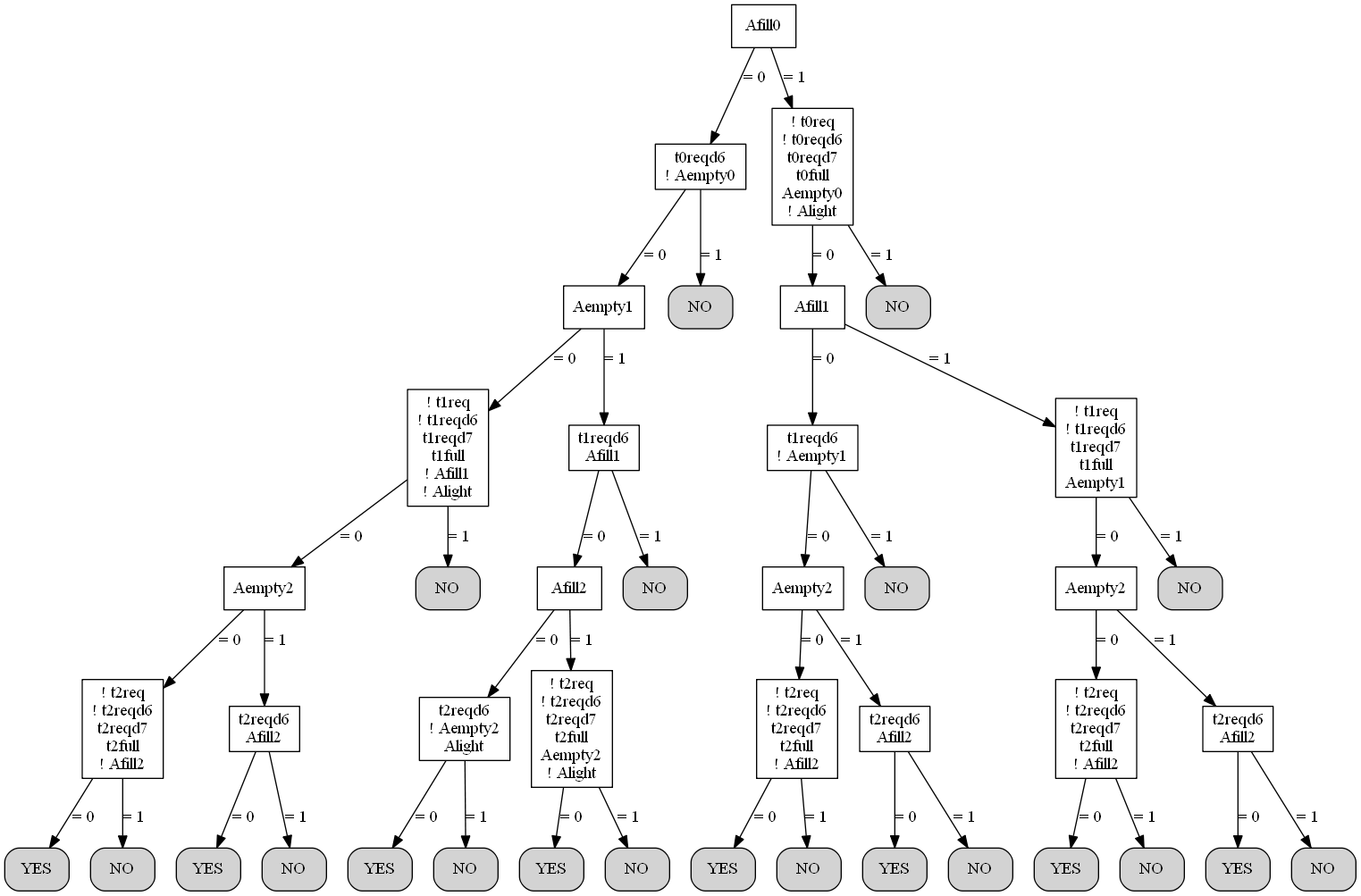}
	\end{center}
\vspace{-3mm}
\caption{A solution for three tanks and empty delay of one, illustration for fill delay of $7$.}
\label{fig:washsafeuniv3tanks}
\vspace{4mm}
\end{figure}


Moreover, unlike BDDs, the computed decision trees provide a scalable universal solution for the whole family of benchmarks.
Fig.~\ref{fig:bstree} shows the decision tree computed for the benchmark with the lowest parameter value and the highest
parameter value, respectively. The scalable universal solution is more apparent and easier to understand once we use the chain
heuristic described in Section~\ref{subs:heuristics} to construct the decision trees. Fig.~\ref{fig:bstreeplus} shows the two decision trees
once the chain heuristic is enabled.

\section{Details of Section~\ref{subs:experrandomltl}: Random LTL}\label{app:randltl}

Fig.~\ref{fig:ltlbasicchained} plots the ratios of decision tree sizes with and without the chain heuristic described in
Section~\ref{subs:heuristics}. In 21 cases the size of both trees is the same, in the remaining 955 cases
the decision tree becomes smaller after applying the heuristic. All three types of the average ratio are around $52\%$.

\begin{figure}
	\begin{center}
	\includegraphics[scale=0.44]{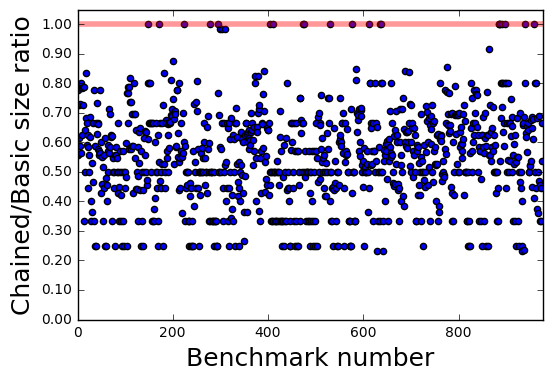}
	\end{center}
\vspace{-4mm}
\caption{Basic vs Chained decision trees}
\label{fig:ltlbasicchained}
\end{figure}

\begin{figure}
	\begin{minipage}{.9\textwidth}
	\begin{center}
		\includegraphics[scale=0.3]{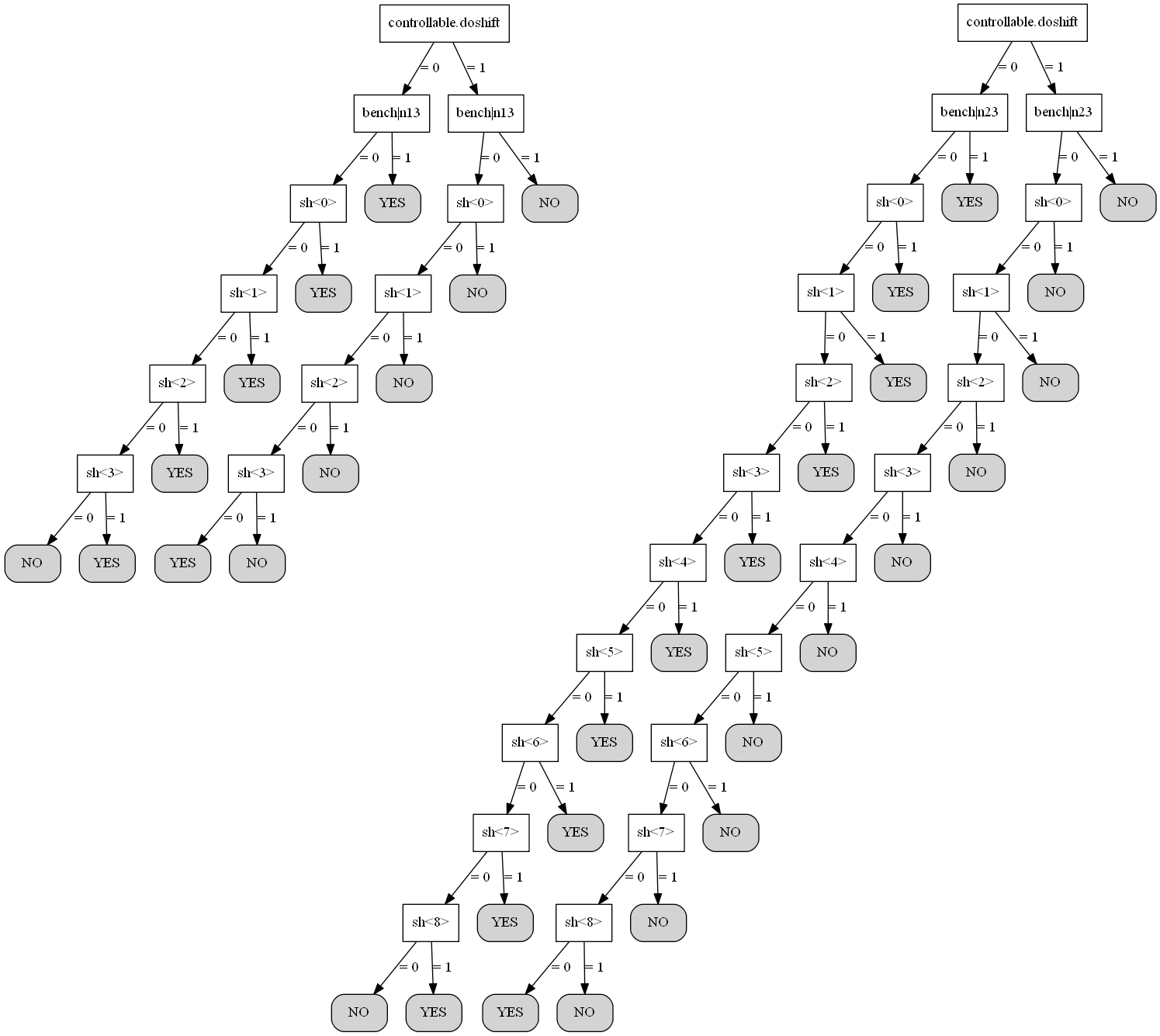}
	\end{center}
	\caption{bs16n decision tree and bs512n decision tree}
	\label{fig:bstree}
	\vspace{6mm}
	\end{minipage}
	\begin{minipage}{.9\textwidth}
	\begin{center}
		\includegraphics[scale=0.4]{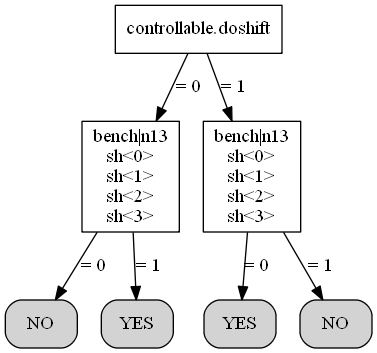}
		\includegraphics[scale=0.4]{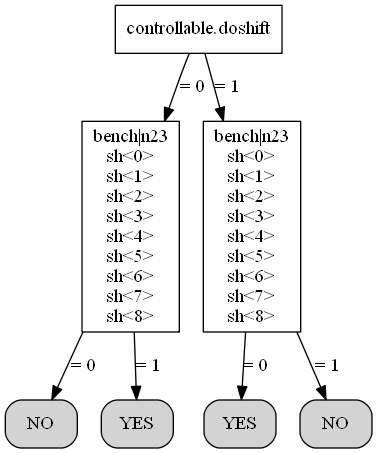}
	\end{center}
	\caption{bs16n decision tree and bs512n decision tree when constructed with the chain heuristic}
	\label{fig:bstreeplus}
	\end{minipage}
\end{figure}

\section{Details of Section~\ref{subs:experrandomltl}: From LTL to I/O Games}\label{app:parity}

\smallskip\noindent{\bf Objective transformation.}
LTL formulae can be translated into non-deterministic B\"uchi automata~\cite{VardiWolper86lics},
and then translated to deterministic parity automata~\cite{Safra88}. 
The synchronous product of the game graph and deterministic parity automata thus
gives rise to graph games with parity objectives.

While such translation is doubly exponential in the worst-case, practically this is rarely the case and there are
efficient tools \cite{spot,rabinizer} allowing to translate some reasonably sized formulae.
Moreover, the number of priorities can also be limited.
For instance, the GR(1) fragment can be translated to 
parity automata with three priorities. 

The first conclusion is that the resulting parity automaton corresponds to the arena of the game and each
state of the automaton can be encoded into binary, resulting in a sequence of state variables for the I/O
game with variables (see the main body of the text). The second conclusion is that we are interested in
positional strategies in these games since parity games allow for memoryless winning strategies, too.

It remains to show how to solve the parity games in our setting efficiently.

\vspace{4mm}
\smallskip\noindent{\bf Strategy construction in parity games.}
There are several algorithms to solve parity game and several solvers
available \cite{DBLP:conf/atva/FriedmannL09,DBLP:conf/atva/MeyerL16}.
Here we use the classical algorithm of Zielonka, tailored to parity 3, covering such
fragments as GR(1) in polynomial time. The algorithm is recursive.
Consider that~0 is the least priority in the game.
Let $\varphi$ and $\overline{\varphi}$ denote the parity objectives of 
player~1 and player~2, respectively. 
The algorithm repeats the following steps:
\begin{compactenum}
	\item The algorithm first computes the set $Y$ of states such that player~1 can 
	ensure reach the set of states with priority~0.
	
	\item Consider the subgame $G'$ without the set $Y$ of states (which has one less priority).
	The subgame $G'$ is solved recursively.
	Let $Z=W_2(G',\overline{\varphi})$ denote the winning region for player~2 in the subgame.
	
	\item The set of vertices in the original game such that player~2 can ensure to reach 
	$Z$ is removed as part of winning region for player~2, and then the algorithm 
	repeats the above steps on the remaining game graph. 
\end{compactenum}
The algorithm stops when $Z$ is empty, and then the remaining states represent
the winning region $W_1(G,\varphi)$ for player~1.
The winning strategy construction in such games is obtained from winning strategies
for reachability objectives and safety objectives. 
An explicit construction of winning strategies in parity games from winning 
strategies for reachability and safety objectives is presented in~\cite{BCWDH10}
(even in the context of partial-information games).

\end{document}